\documentclass[preprint,12pt,authoryear,review,3p]{elsarticle}

\usepackage{amssymb}
 \usepackage{amsthm}

\usepackage{lineno}
\usepackage{amsmath}

\providecommand{\bo}{\mathbf}
\providecommand{\bs}{\boldsymbol}
\providecommand{\cov}{\mathrm{COV}}
\providecommand{\mlt}{\mathrm{MLC}}
\providecommand{\mcd}{\mathrm{MCD}}
\providecommand{\diag}{\mathrm{diag}}
\providecommand{\var}{\mathrm{var}}
\newtheorem{prop}{Proposition}
\newtheorem{rem}{Remark}
\newtheorem{coro}{Corollary}

\journal{Computational Statistics \& Data Analysis}

\begin{document}

\begin{frontmatter}



\title{ICS for Multivariate Outlier Detection\\ with Application to Quality Control}

\author{Aurore Archimbaud\fnref{fn1}\fnref{ad1}}
\ead{aurore.archimbaud@tse-fr.eu}
\author{Klaus Nordhausen\fnref{ad2}}
\ead{klaus.nordhausen@tuwien.ac.at}

\author{Anne Ruiz-Gazen\fnref{ad1}\corref{cor1}}
\ead{anne.ruiz-gazen@tse-fr.eu}
\cortext[cor1]{Corresponding author}

\address[ad1]{Toulouse School of Economics, University of Toulouse 1 Capitole,  \\ 21 all\'ee de Brienne, 31015 Toulouse cedex 6, France}
\address[ad2]{CSTAT - Computational Statistics, Institute of Statistics \& Mathematical Methods in Economics
Vienna University of Technology, Wiedner Hauptstr. 7, A-1040 Vienna, Austria}

\begin{abstract}
In high reliability standards fields such as automotive, avionics or aerospace, the detection of anomalies is crucial.
An efficient methodology for automatically detecting multivariate outliers is introduced.
It takes advantage of the remarkable properties of the Invariant Coordinate Selection (ICS) method.
Based on the simultaneous spectral decomposition of two scatter matrices, ICS leads to an affine invariant coordinate system
in which the Euclidian distance corresponds to a Mahalanobis Distance (MD) in the original coordinates.
The limitations of MD are highlighted using theoretical arguments in a context where the dimension of the data is large.
Unlike MD, ICS makes it possible to select relevant components which removes the limitations.
Owing to the resulting dimension reduction, the method is expected to improve the power of outlier detection rules such as MD-based criteria.
It also greatly simplifies outliers interpretation.
The paper includes practical guidelines for using ICS in the context of a small proportion of outliers
which is relevant in high reliability standards fields.
The choice of scatter matrices together with
the selection of relevant invariant components through
parallel analysis and normality tests are addressed.
The use of the regular covariance matrix and the so called matrix of fourth moments as the scatter pair is recommended.
This choice combines the simplicity of implementation together with the possibility to derive theoretical results.
A simulation study confirms the good properties of the proposal and compares it with other scatter pairs.
This study also provides a comparison with Principal Component Analysis and MD.
The performance of our proposal is also evaluated on several real data sets using a user-friendly R package accompanying the paper.
\end{abstract}


\end{frontmatter}

%\begin{linenumbers}

%%\section{}
%%\label{}



\section{Introduction}
\label{sec:intro}

Detecting outliers in multivariate data sets is of particular interest in many physical \citep{beckman1983outlier}, industrial, medical and financial applications \citep{aggarwal2017}.
Some classical statistical detection methods are based on the Mahalanobis distance and its robust counterparts (see e.g. \citet{rousseeuw1990unmasking}, \citet{cerioli2009controlling}, \citet{cerioli2010multivariate}) or on robust principal component analysis (see e.g \citet{hubert2005robpca}).
One advantage of the Mahalanobis distance is its affine invariance while Principal Component Analysis (PCA) is only invariant under orthogonal transformations.
For its part, PCA allows some components selection and facilitates the interpretation of the detected outliers.
All these methods are adapted to the context of casewise contamination while other methods are adapted to the case of cellwise contamination (see e.g. \citet{agostinelli2015robust} and \citet{rousseeuw2016detecting}). Furthermore, several other recent references tackle the problem of outlier detection in high dimension where the number of observations may be smaller than the number of variables (see e.g. \citet{croux2013robust} and \citet{hubert2016sparse}).

In the present paper, we propose an alternative to the Mahalanobis distance and to PCA, in a casewise contamination context and when the number of observations is larger than the number of variables.
As stated in \cite{tarr2016} on page 405: ``the cellwise contamination is prevalent in large, automatically generated data sets,
found in data mining and bioinformatics, where there is often little quality control over the inputs''. In the present paper, the focus is on applications with high level of quality control, such as in the automotive, avionics or aerospace fields, where only a small proportion of outliers, up to 2\%, is plausible.
From our experience in such application fields, a small proportion of parts potentially  defective are to be detected with very limited false detection.
Moreover, even if in such fields the trend is to increase the number of measurements, there are still many applications where the number of observations is larger than the number of variables and, in such a context, an improved affine invariant method with an easy characterization of the outliers is still of interest.

The method we consider is the Invariant Coordinate Selection (ICS) as proposed by \citet{Tyler2009}.
The principle of ICS is quite similar to Principal Component Analysis (PCA) with coordinates or components derived from an eigendecomposition followed by a projection of the data on
selected eigenvectors. However, ICS differs in many respects from PCA. It relies on the simultaneous spectral decomposition of two scatter matrices instead of one for PCA.
While principal components are orthogonally invariant but scale dependent, the invariant components are affine invariant for affine equivariant scatter matrices.
Moreover, under some elliptical mixture models, the Fisher's linear discriminant subspace coincides with a subset of invariant components
in the case where group identifications are unknown (see Theorem 4 in \cite{Tyler2009}).  This remarkable property is of interest for outlier detection since
outliers can be viewed as data observations that differ from the remaining data and form separate clusters.

Despite its attractive properties, ICS has not been extensively studied in the literature on outlier detection.
An early version of ICS was proposed in \citet{caussinus1990interesting} for multivariate outlier detection and studied further in e.g. \citet{penny1999multivariate}
and \citet{caussinus2003projections} for two specific scatter matrices.
Recent articles by \citet{ICS} and \citet{Tyler2009} argue that ICS is useful for outlier detection. However, a thorough evaluation of ICS in
this context is still missing and the present paper is a first step aimed at filling the gap.

Our first objective is to explain the link between ICS and the Mahalanobis distance. First, we prove that Euclidian distances calculated using all invariant components
are equivalent to Mahalanobis distances calculated using the original variables. Then, in the case where the number of variables is large (but still with a larger number of observations) and outliers are contained in a
small dimensional subspace, we recommend selecting a small number of invariant components. Such a selection is motivated by looking at the approximate probability  in large
dimension of the difference between the Mahalanobis distance of an outlying observation and the Mahalanobis distance of an observation from the majority group. We prove
that this probability decreases toward zero when the dimension increases which is undesirable. This shortcoming can be avoided by a proper selection of invariant components.

Then, we focus on the case where the majority of the data behaves in a regular way and only a small fraction of the data might be  considered outliers.
Examples include, for instance, financial fraud detection or production error identification in industrial processes where there is a high level of quality control.
Our goal is to provide practical guidelines for using ICS in this context of unsupervised detection of a small proportion of outliers.
More precisely, we implement and compare different pairs of scatter matrices estimators and different methods for selecting relevant invariant components through
an extensive simulation study. We consider several contamination models with a percentage of contamination equal to 2\%, which is relevant in the context of high reliability standards fields.
Results are given in terms of true positive and false negative discoveries for several mixture models.
We advocate a simple choice for the scatter matrices pair, namely the covariance and the fourth moment matrices.
Such estimators are simple to implement and some theoretical results can be derived for some particular mixtures as detailed in \ref{sec:CompChap2}.
Regarding components selection, we recommend two methods: the so-called parallel analysis
\citep{PeresNetoJacksonSomers2005} and a skewness-based normality test.
We also show that our proposal improves over the Mahalanobis distance criterion and over different versions of PCA through simulations and the use of three real data sets.
One of the key benefits of our approach compared to competitors is its ability not to detect outliers when there is no outlier present in the data set, at least in the Gaussian case.
When outliers are absent, the proposed procedure is likely to select none of the invariant components. Another practical benefit, as illustrated on one of the three real examples,
is the ease of interpretation of the detected outliers using the selected invariant coordinates. Mimicking PCA, the user can draw some scatter plots of the invariant components or
look at the correlations between the invariant components and the original variables. More complex procedures (advocated for instance in \citet{willems2009}) when using the
Mahalanobis distance can thereby be avoided.

This article is organized as follows. In Section \ref{sec_mdld} we observe the behavior of the usual and the robust Mahalanobis distances for large dimensions when
outliers lie in a small dimensional subspace. This result motivates the use of selected invariant components for outlier detection. ICS is described in a general
framework in Section \ref{sec_ics} and in the context of a small proportion of outliers in Section \ref{sec_icsout}. Section \ref{sec_simu} provides results from a
simulation study and derives practical guidelines for the choice of the scatter matrices pair and the components selection method. Comparisons with the Mahalanobis distance
and PCA are also provided. Three real data sets are analyzed in Section \ref{sec_real}.
Finally, conclusions and perspectives are drawn in Section \ref{sec_concl}. The proof of Proposition \ref{propmd} is given in \ref{proofprop1} while some additional propositions are given in \ref{sec:CompChap2}.
Supplementary material is also provided. It contains  some scatterplot matrices to visualize the six simulated data sets and the R code to generate these data sets. It also includes
the R code to reproduce the results of Table 4 for the Reliability data and the HTP data sets.

\section{Behavior of the Mahalanobis distance in large dimension}
\label{sec_mdld}
Let $\bo X=(X_1,\ldots,X_p)'$ be a $p$-multivariate real random vector and assume the distribution of $\bo X$ is a mixture of $(q+1)$ Gaussian
distributions with $q+1 < p$, different location parameters $\bs\mu_h$, for $h=0,\ldots,q$, and the same definite positive covariance matrix $\bs\Sigma_W$:
\begin{equation}
{\bo X} \sim (1-\epsilon) \, {\cal N}(\bs\mu_0,\bs\Sigma_W) + \sum_{h=1}^q \epsilon_h \, {\cal N}(\bs\mu_h,\bs\Sigma_W)
 \label{model}
\end{equation}
where $\displaystyle \epsilon =  \sum_{h=1}^q \epsilon_h < 1/2$.

Such a distribution can be interpreted as a model for outliers where the majority of the data follows a given Gaussian distribution and outliers are clustered in $q$ clusters with
Gaussian distributions with different locations than the majority group. This model is a generalization of the well-known mean-shift outlier model to more than two groups.

For such a model, the mean is $\bs\mu_{\bo X} = (1-\epsilon) \, \bs\mu_0 + \sum_{h=1}^q \epsilon_h \bs\mu_h$, the within covariance matrix is $\bs\Sigma_W$, the between covariance is
$\bs\Sigma_B = (1-\epsilon)(\bs\mu_0-\bs\mu_{\bo X})(\bs\mu_0-\bs\mu_{\bo X})' + \sum_{h=1}^q \epsilon_h (\bs\mu_h-\bs\mu_{\bo X})(\bs\mu_h-\bs\mu_{\bo X})'$,
where the prime symbol denotes the transpose vector or matrix, and the total covariance matrix is
$\bs\Sigma=\bs\Sigma_B + \bs\Sigma_W$.
Let us consider the following squared Mahalanobis distances:
\begin{eqnarray}
d^2(\bo X) & = & (\bo X-\bs\mu_{\bo X})'{\bs\Sigma}^{-1}(\bo X-\bs\mu_{\bo X})\\
d_R^2(\bo X) & = & (\bo X-\bs\mu_0)'{\bs\Sigma_W}^{-1}(\bo X-\bs\mu_0).
\end{eqnarray}
These distances are affine invariant in the sense that $d^2(\bo A \bo X + \bo b)=d^2(\bo X)$ and $d_R^2(\bo A \bo X + \bo b)=d_R^2(\bo X)$,
for  any full rank $p\times p$ matrix $\bo A$ and  any $p$-vector $\bo b$. The distance $d$ (resp. $d_R$) can be interpreted as a non-robust (resp. robust) Mahalanobis distance. Of course in practice, the different parameters are unknown and
should be estimated, but the results we derive below give some intuition for the finite sample case.
Let us now introduce distinct $p$-random vectors that would correspond to the different mixture components of $\bo X$. Let $\bo X_{no}$, where $no$ stands for ``non-outlier'',
follows a normal distribution ${\cal N}(\bs\mu_0,\bs\Sigma_W)$ and $\bo X_{o,h}$, where
$o$ stands for ``outlier'', follows a normal distribution
${\cal N}(\bs \mu_h,\bs\Sigma_W)$, with  $h=1,\ldots,q$. We assume that $\bo X_{no}$ and $\bo X_{o,h}$,  for $h=1,\ldots,q$,
are independent, and we are interested in the behavior of the difference between the squared distance of $\bo X_{o}$ and of $\bo X_{no}$ for both Mahalanobis distances,
when dimension $p$ increases. The distribution of these differences is not easy to handle especially for the non-robust distance, but we can look at the asymptotic
distribution for large $p$.
When using the Mahalanobis distance or robust distance for outlier identification, we expect the probability of these differences to be large.

Under the mixture distribution defined previously, we have the following proposition. Its proof makes use of the Lindeberg-Feller central limit theorem for $p$ going to infinity and is given in \ref{proofprop1}.\\

\begin{prop}\label{propmd}
Assume that $q$ is fixed, then under model (\ref{model}):
$$\frac{1}{2\sqrt{p}} \left(d^2(\bo X_{o,h}) -d^2(\bo X_{no}) - \mathbb{E}\left( d^2(\bo X_{o,h}) -d^2(\bo X_{no}) \right) \right)$$
and
$$ \frac{1}{2\sqrt{p}} \left(d_R^2(\bo X_{o,h}) -d_R^2(\bo X_{no})- \mathbb{E}\left( d_R^2(\bo X_{o,h}) -d_R^2(\bo X_{no})\right)  \right)$$
converge in distribution to a standard Gaussian distribution when $p$ goes to infinity and the expectations
$\mathbb{E}\left( d^2(\bo X_{o,h})\right.$ $\left.-d^2(\bo X_{no}\right) $ and $\mathbb{E}\left( d_R^2(\bo X_{o,h}) -d_R^2(\bo X_{no})\right)$ do not depend on $p$.
\end{prop}

Note that under  model (\ref{model}), the expectations can be made explicit but their expressions are complex and not detailed further.
The conclusion of Proposition 1 is that if outliers belong to a reduced dimension space (equal at most to $q$ in model (\ref{model})) and $p$ is large, then the probability that the Mahalanobis distance of an outlier exceeds the Mahalanobis distance of a non-outlier is small, because according to the asymptotic result, the variance of the differences increases when $p$ increases.
This makes the outlier identification more difficult. If the $q$-subspace is known, it is easy to avoid the problem of the $p-q$ noisy dimensions by projecting the data set
on this subspace and calculating a distance based on the $q$ dimensions that does not depend on $p$. This is exactly what ICS is all about, providing the data-analyst with
the ability to select a subspace displaying the outliers in an unsupervised way, and project the data on this subspace. Figure \ref{fig:MD1} illustrates in some sense Proposition 1
results and the competitive advantage of ICS compared to the Mahalanobis distance on a simple artificial data set. This set which will be discussed in the simulation
framework as ``Case 1'', contains 1000 observations with one cluster of 20 outliers  location shifted and plotted in black. The dimension of the data set increases from $p=6$
on the left panels, to 25 on the middle ones and 50 on the right panels. The top panels plot the non-robust Mahalanobis distances using the usual covariance estimator
while the middle panels plot robust Mahalanobis distances using the (reweighted) MCD estimator \citep{Rousseeuw1986}. The bottom panels plot the distances based on an automatic selection
of invariant components for ICS with a pair of scatter matrices estimators detailed later in the present paper.
When $p$ increases, it becomes more difficult to separate the outlying observations from the rest of the data using the Mahalanobis distances while the separation remains
much better using selected invariant components.
ICS is now detailed, and the choice of the scatter pair together with the selection of the invariant components is discussed.

\begin{figure}
   \caption{Squared distances (top: non-robust Mahalanobis, middle: robust Mahalanobis, bottom: Euclidian using invariant components with an automatic selection) for $p=6$ (resp. 25 and 50)
  on the left (resp. middle and right) panels for a sample of 1000 observations drawn from a mixture of two normal distributions with the 20 location shifted observations
  in black.}\label{fig:MD1}
\includegraphics[width=0.99\textwidth]{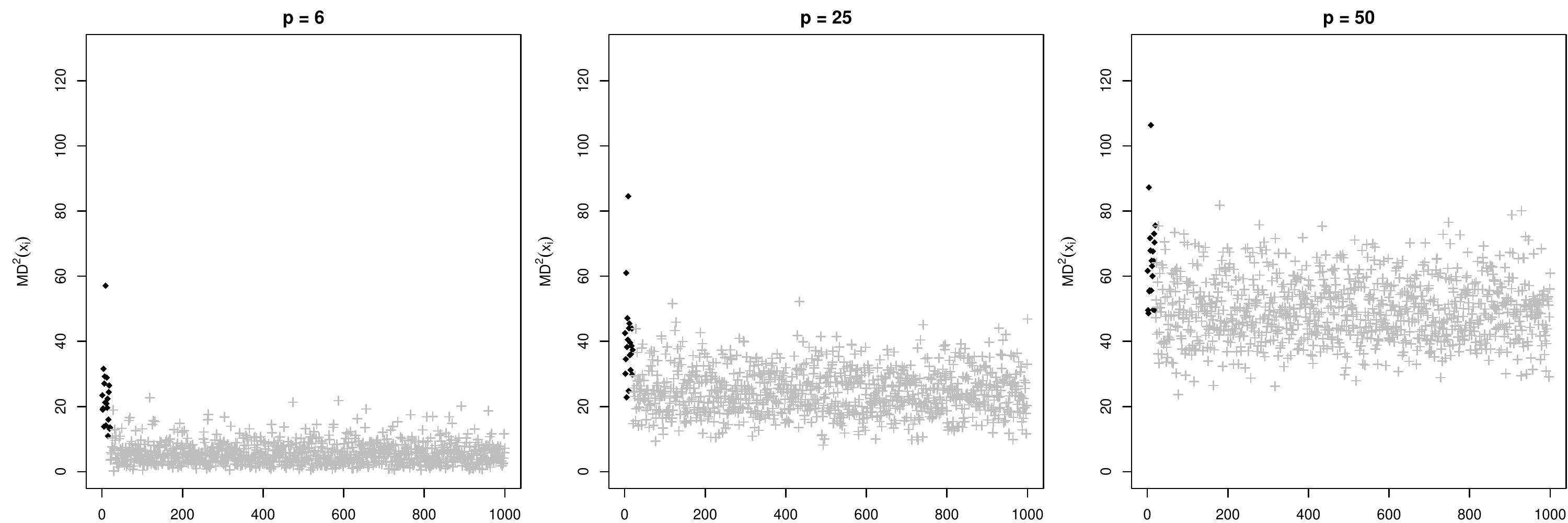}\\
\includegraphics[width=0.99\textwidth]{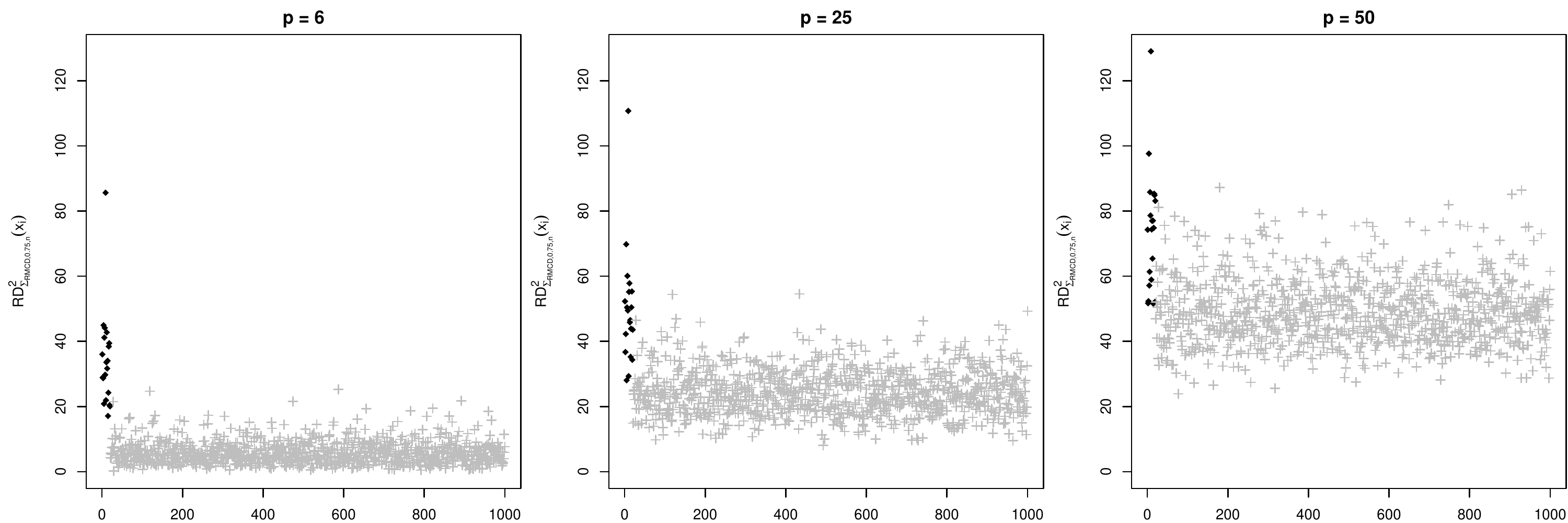}\\
\includegraphics[width=0.99\textwidth]{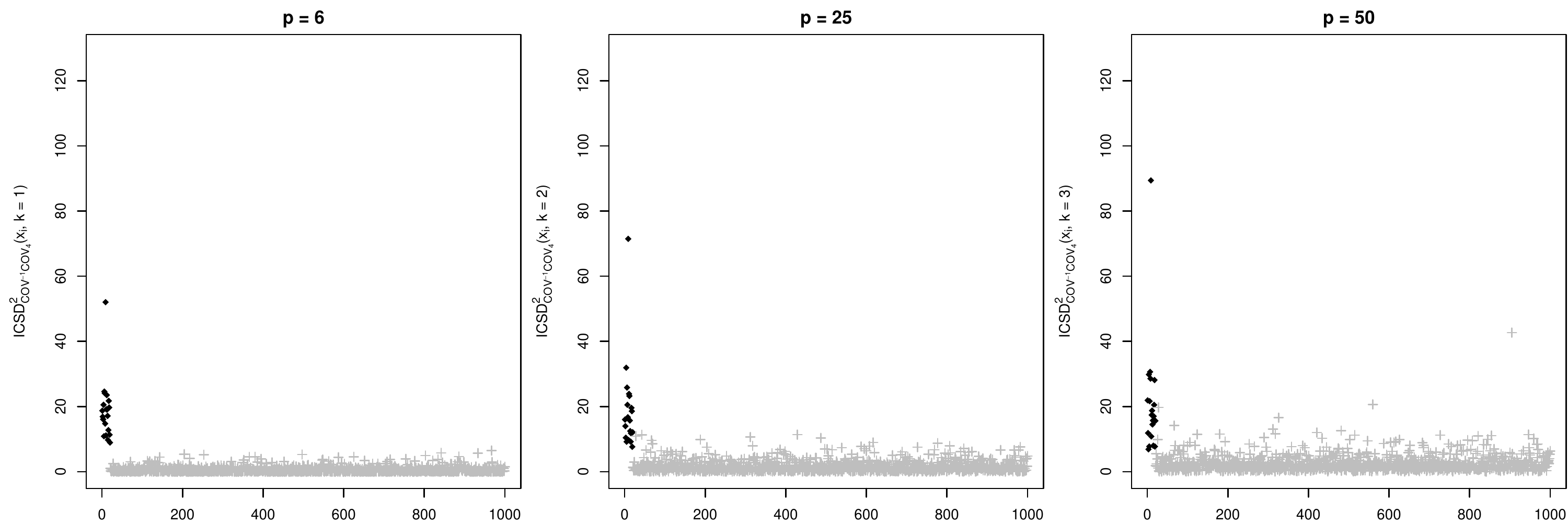}\\
\end{figure}

\section{Invariant Coordinate Selection}
\label{sec_ics}

\subsection{Scatter matrices}

For a $p$-variate dataset $\bo X_n=(\bo x_1,\ldots,\bo x_n)'$, any $p \times p$ matrix symmetric and definite positive $\bo V(\bo X_n)$ is a scatter matrix if it is affine equivariant in the sense that
\[
\bo V(\bo X_n \bo A + \bo 1_n \bo b') = \bo A' \bo V(\bo X_n) \bo A,
\]
where $\bo A$ is a full rank $p \times p$ matrix, $\bo b$ a $p$-vector and $\bo 1_n$ an $n$-vector full of ones.

The literature contains numerous scatter matrices suggestions (see \citet{NordhausenTyler2015} for a recent discussion and
many references).
\citet{Tyler2009} classify them into three classes depending on their robustness properties in terms of breakdown point and influence function.
Class I scatter matrices have a zero or almost zero breakdown value and an unbounded
influence function. Relevant scatter matrices from this class are the regular covariance matrix
\[
\cov(\bo X_n)= \frac{1}{n-1} \sum_{i=1}^{n} (\bo x_i-\bar{\bo x})(\bo x_i-\bar {\bo x})',
\]
where $\bar {\bo x}$ denotes the empirical mean, and the so called scatter matrix of fourth moments
\[
\cov_4(\bo X_n) = \frac{1}{(p+2)n} \sum_{i=1}^{n} r_i^2 (\bo x_i-\bar {\bo x})(\bo x_i-\bar {\bo x})',
\]
where $r_i^2 = (\bo x_i - \bar {\bo x})'\cov(\bo X_n)^{-1}(\bo x_i - \bar {\bo x})$ is the classical squared Mahalanobis distance.

Class II consists of scatter matrices with a bounded influence function but a breakdown point not larger than $(p+1)^{-1}$.
From this class, we will later use the following location and scatter matrix estimators defined through the implicit expressions:
\begin{eqnarray*}
\bo m_C(\bo X_n) & = & \sum_{i=1}^{n}(w(r_i^2)\bo x_i)/\sum_{i=1}^{n} w(r_i^2),\\
\mlt(\bo X_n)& = &\frac{1}{n} \sum_{i=1}^n w(r_i^2) (\bo x_i - \bo m_C(\bo X_n))(\bo x_i- \bo m_C(\bo X_n))',
\end{eqnarray*}
where  $r_i^2= (\bo x_i-\bo m_C(\bo X_n))'\mlt(\bo X)^{-1}(\bo x_i-\bo m_C(\bo X_n)$ and
$w(r_i^2) = (p+1)/(r_i^2+1)$.
These location and scatter matrix estimators are the maximum likelihood estimators of an elliptical Cauchy distribution and belong to the well-known class of M-estimators.

Class III scatter matrices are high-breakdown scatter matrices, and the reweighted Minimum Covariance Determinant (MCD) estimator is perhaps the most popular example from this class.
For a given $h \in [0.5; 1]$, the $\mcd_h$ searches for the $hn$ observations $\bo X_{hn}$ such that
$\cov(\bo X_{hn})$ has the smallest determinant and then is made more efficient by reweighting observations appropriately (see
\cite{Rousseeuw1986} and \cite{cator2012} for more details). The associated location estimator is a reweighted version of the average of the $hn$ observations.

While the Mahalanobis distance and PCA are based on one scatter matrix, ICS is based on the simultaneous use of two scatter matrices denoted below by
$\bo V_{1}(\bo X_n)$ and $\bo V_{2}(\bo X_n)$. We will choose among the four estimators recalled previously and consider that class III
estimators are more robust than class II, which are themselves more robust than class I. For the two class I estimators $\cov(\bo X_n)$ and $\cov_4(\bo X_n)$,
we will consider $\cov(\bo X_n)$ more robust than $\cov_4(\bo X_n)$ because the norm of its influence function is smaller.

% and where the convention is in this paper to choose as $\bo V_{1}$ always the less robust scatter matrix.

\subsection{ICS principle}

Formally, the goal of ICS is to find the $p \times p$  matrix $\bo B(\bo X_n)$ and diagonal matrix $\bo D(\bo X_n)$ such that:
\[
\bo B(\bo X_n) \bo V_{1}(\bo X_n) \bo B(\bo X_n)' = \bo I_p  \;\; \mbox{ and}
\]
\[
\bo B(\bo X_n) \bo V_{2}(\bo X_n) \bo B(\bo X_n)' = \bo D(\bo X_n).
\]
$\bo D(\bo X_n)$ contains the eigenvalues of $\bo V_{1}(\bo X_n)^{-1} \bo V_{2}(\bo X_n)$ in decreasing order,
 while the rows of the matrix $\bo B(\bo X_n)=(\bo b_1,\ldots,\bo b_p)'$ contain the corresponding eigenvectors so that:
\[
\bo V_{1}(\bo X_n)^{-1} \bo V_{2}(\bo X_n) \bo B(\bo X_n)' = \bo B(\bo X_n)' \bo D(\bo X_n).
\]
Using any affine equivariant location estimator $\bo m(\bo X_n)$, the corresponding scores
\[
\bo Z_n =(\bo z_1,\ldots,\bo z_n)'= (\bo X_n- \bo 1_n \bo m(\bo X_n)')\bo B(\bo X_n)'
\]
are the so-called invariant coordinates or components. They are affine invariant in the sense that
\[
(\bo X_n^* - \bo 1_n \bo m(\bo X_n^*)')\bo B(\bo X_n^*)' =  (\bo X_n - \bo 1_n \bo m(\bo X_n)')\bo B(\bo X_n)' \bo J
\]
for $\bo X_n^* = \bo X_n \bo A+ \bo 1_n \bo b'$ with any full rank $p \times p$ matrix $\bo A$ and any $p$-vector $\bo b$. $\bo J$
is a $p \times p$ diagonal matrix with diagonal elements $\pm 1$, which means the invariant coordinates change at most their signs.
For convenience, the dependence on $\bo X_n$  is dropped from the different matrices when the context is obvious.

\noindent Because $\bo V_1^{-1}=\bo B' \bo B$, the proof of the following proposition is immediate.

\vspace{2mm}

\begin{prop}
Let us consider an affine equivariant location estimator $\bo m$ and two scatter matrices $\bo V_1$ and $\bo V_2$.
The Euclidian norm of an observation using its invariant coordinates corresponds to the Mahalanobis distance
of this observation from $\bo m$ in the sense of $\bo V_1$. Formally, it means that for observation $i=1,\ldots,n$,
\[
\bo z_i' \bo z_i= (\bo x_i- \bo m)'\bo V_1^{-1}(\bo x_i- \bo m)
\]
\end{prop}
\cite{Tyler2009}, p. 554, underlines the exchangeability between the roles of $\bo V_1$ and $\bo V_2$. However, as can be observed from Proposition 2,
exchanging the two scatter matrices has an impact on the scale of the invariant coordinates and not only on the fact that the eigenvalues are the inverse of the
others and the eigenvectors are in reverse order. In the following, we propose to use the location estimator associated with the scatter matrix $\bo V_1$  and take $\bo V_1$   ``more'' robust than $\bo V_2$, as \citet{Alashwali2016}, so that the Euclidian distance using all invariant components leads to a more robust Mahalanobis distance.

\section{ICS implementation for outlier detection}
\label{sec_icsout}

Identifying outliers with ICS is a three step procedure. The first step consists in choosing a pair of scatter matrices and calculating the
invariant coordinates. The second step is the selection of the relevant invariant components and the calculation of the Euclidian norm of the $n$ observations
using only the selected components.
The last step is the outlier identification with the choice of a cut-off value $c$ such that observations with a norm larger than $c$ are flagged as outliers.

\subsection{The choice of the scatter pair}

When the objective is outlier detection, Caussinus and Ruiz-Gazen (1990, 2003) and \citet{Tyler2009} recommend using class I scatter estimators such as the classical
one or some weighted scatter matrix. The main reason for this choice is that these estimators are simple and can be computed rapidly.
Moreover, the nice properties of ICS given by Theorems~3 and 4 in \citet{Tyler2009} are true even for non-robust estimators such as the  $\cov$ and  $\cov_4$ scatter matrices. For this particular pair, the formulation of Theorem 3, which applies to a mixture of two Gaussian distributions with different locations and proportional scatter matrices, can be made much more precise.
As explained in \citet{Tyler2009}, for a proportion of outliers smaller than $(3-\sqrt{3})/6$ (around 21\%), the first invariant component displays the outliers.
Similar results can be derived for other particular mixtures as detailed in \ref{sec:CompChap2}.
More precisely, for a symmetrically contaminated Gaussian distribution with equal covariance matrices (which is similar to the so-called barrow wheel distribution), the first component will display the structure as soon as the contamination level is smaller than 33\%. And this is also true for a Gaussian mixture with inflated variance in $q$ directions: as soon as the contamination is smaller than 50\%, the invariant components associated with the $q$ largest eigenvalues will span the subspace of interest.
For other scatter pairs, this calculus is not analytically tractable anymore and so a comparison through simulations is worthwhile.
In the present paper, we propose comparing four pairs of scatter matrix estimators taken from the three different classes based on simulations.
The first pair is based on two class I estimators $\bo V_1 = \cov$ and $\bo V_2=\cov_4$, while the others are based on class II and I with $\bo V_1 = \mlt$ and $\bo V_2=\cov$, class III and I with $\bo V_1 = \mcd$ and $\bo V_2=\cov$ and class III and II scatter estimators with $\bo V_1 = \mcd$ and $\bo V_2=\mlt$.

\subsection{The invariant components selection}
\label{subsec:ICSselect}

In the present subsection, we focus on the case of a small proportion of outliers that could be as high as 20\% if we take into account the theoretical properties of ICS for the $\cov - \cov_4$ pair as detailed in \ref{sec:CompChap2}.
We assume that the outliers belong to a subspace of dimension $q\leq p$, and we aim at providing some procedures to automatically select a number of invariant components close to $q$.
Beginning with the first component, we test whether each invariant component is significantly relevant via two different sequential approaches. For both approaches, as soon as one invariant component, - let us say number ($k+1$), -  is not significantly relevant, we stop the procedure and select the $k$ first components. In this particular context of sequential multiple testing, some adjustments on the initial significance level $\alpha$ are necessary. Following \citet{Dray2008}, we apply the Bonferroni correction on the significance level and consider a level $\alpha_j=\alpha/j$ for each component $j=1,\ldots,p$.

The first approach consists in a Parallel Analysis (PA) based on Monte Carlo simulations. For some given dimensions $n$ and $p$, many samples are generated following a standard multivariate Gaussian distribution, and for each sample and a given scatter pair, the eigenvalues of the simultaneous diagonalization of the two scatter matrices are computed. Cut-offs for the eigenvalues are then derived using the empirical quantiles of the eigenvalues from the simulated Gaussian data. This method is common for selecting components in PCA as described in \citet{PeresNetoJacksonSomers2005}. It was already used in \citet{caussinus2003projections} for ICS but only for a particular pair of scatters.
The second approach makes use of the fact that relevant components for displaying outliers do not follow a Gaussian distribution. It is thus based on univariate normality tests for each component beginning with the first one as previously described. The five tests we compare are the D'Agostino test of skewness (DA), the Anscombe-Glynn (AG) test of kurtosis, the Bonett-Seier (BS) test of Geary's kurtosis, the Jarque-Bera (JB) test based on both skewness and kurtosis and the Shapiro-Wilk (SW) normality test (see \citet{Yazici2007} and \citet{bonett2002test} for a complete description of these five tests).

Note that automated selection procedures are necessary in a simulation framework but may not be the best alternative when analyzing one data set. This point will be detailed further in the data analysis section of the present paper, where we also explore the possibility of using a scree plot as in PCA.

%As in the first approach we test the first $i^{th}$ coordinates at the $\alpha_i$ levels and we continue until we cannot reject anymore. \\

\subsection{Outlier identification}
Once having selected $k$ invariant components, the last  procedure step is the identification of outlying observations.
For each observation $i=1,\ldots,n$, we calculate its squared ``ICS distance'' which corresponds to its squared Euclidian norm in the invariant coordinate system taking into account
the first $k$ coordinates:
\[
\left(\mbox{ICS distance}\right)^2_i=\sum_{j=1}^k \left(z_i^j \right)^2
\]
where $z_i^j$ denotes the $j$th coordinate of the score $\bo z_i$.
As the distribution of the ICS distances is unknown, we derive cut-offs
based on Monte Carlo simulations from the standard Gaussian distribution.
For a given data dimension, a scatter pair and a number $k$ of selected components, we generate many samples and compute the ICS distances. A cut-off is derived for a fixed level $\gamma$ as the $1-\gamma$ percentile of these distances.
An observation with a distance higher than this cut-off is flagged as an outlier.

The implementation of ICS for outlier detection in the next two sections is performed  in R 3.1.2 \citep{R} using the packages ICS \citep{ICS}, ICSOutlier \citep{ICSOutlier}, mvtnorm \citep{mvtnorm}, moments \citep{moments}, robustX \citep{robustX} and robustbase \citep{robustbase}.

\section{Simulations}\label{sec_simu}
\subsection{Simulation framework}

ICS performance for outlier detection is evaluated through an extensive simulation study in the particular context of a proportion of outliers fixed at 2\%.
As already indicated, this small proportion is consistent with some current practice in industrial applications where the data already meet the standard quality controls and
only a few observations, clearly identified as multivariate outliers, may be disregarded.
The different models we consider are well-known models in the robust statistics literature \citep{hampel1986robust}. Using the $\cov-\cov_4$ scatter pair and for the two components mixture models or the scale-shift model, it is possible to derive some theoretical conditions on the contamination level which insure that the first invariant components point in the directions of the outliers (see \ref{sec:CompChap2} for details).

In this framework, we discuss the impact of the scatter pair together with the components selection strategy and the choice of the cut-off for identifying outliers. Some of the conclusions and recommendations drawn from this study are used as guidelines for the data analysis conducted in Section \ref{sec_real} in different industrial settings.
Concerning the scatter matrices, the four pairs (i) $\cov-\cov_4$, (ii) $\mlt-\cov$, (iii) $\mcd-\cov$ and (iv) $\mcd - \mlt$ are evaluated.
In pairs (ii)-(iv) the scatter matrices come from different classes, while in pair (i) both come from class I.
For the $\mcd$, given that the proportion of outliers is small, the value $h=0.75$ which is often advocated \citep{croux1999} is used throughout the simulations, leading to a 25\% breakdown point.

For each of the six setups, we generate $1000$ samples with sample size $n=1000$ and dimension $p$ equal to 6, 25 and 50.
For all cases, the uncontaminated data follow a Gaussian distribution with mean 0 and covariance matrix $\Sigma_i$, $i=0,\ldots,5$, depending on the setup.
Except for Case 0 which contains no outlier, we generate exactly 20 outliers in each sample so that the proportion of outliers is 2\% in all samples.
We use the notation $\bs e_i$ for the $p$-vector with a one in the $i$th coordinate and zero elsewhere. For each setup we give the dimension $q$ of
the subspace spanned by the outliers. For dimension $p=6$, the figure in the supplementary material gives the scatterplot matrix of the variables for one simulation of each of the six cases in order to visualize easily the structure of the data sets. Some affine transformation could have been performed in order to mask the structure like in \citet{StahelMachler2009} for the barrow wheel distribution. Such transformation has no impact on the MD and ICS results but may change completely the PCA results.

\begin{description}
  \item[Case 0 ($q=0$):] $\bs \Sigma_0=\bo I_p$ with no outlier.
  \item[Case 1 ($q=1$):] $\bs \Sigma_1 = \diag(1,4,\ldots,4)$ with outliers clustered in one direction with distribution ${\cal N}(6 \bs e_1, \bs \Sigma_1)$.
  \item[Case 2 ($q=1$):] $\bs \Sigma_2 = \diag(0.1,1,\ldots,1)$ with outliers following a distribution $H$ such that $\bo h=(h_1,\bo h_2')' \sim H$ means that $h_1 \sim \chi_5$ and $\bo h_2 \sim {\cal N}(\bo 0, 0.2 \bo I_{p-1})$. The data follows the so-called barrow wheel distribution as introduced in \citet{hampel1986robust} and using a slightly modified setting compared to \citet{StahelMachler2009}. No rescaling or rotation has been performed. In any case such transformations have no impact on the ICS results. Hence, outliers are generated along the same direction on both sides of the uncontaminated data cloud.
  \item[Case 3 ($q=2$):] $\bs \Sigma_3 = \diag(1,1,4,\ldots,4)$ with outliers clustered in two directions with 12 (resp. 8) observations following a ${\cal N}(6 \bs e_1, \bs \Sigma_3)$
  (resp. ${\cal N}(6.2 \bs e_2, \bs \Sigma_3)$) distribution.
  \item[Case 4 ($q=6$):] $\bs \Sigma_4=\bo I_p$ with outliers clustered in six directions with Gaussian distribution with mean  $\bs \mu_{i} = (6 + 0.1(i-1)) \bs e_i$, $i=1,\ldots,6$ and covariance $\bo I_p$, with 4 (resp. 3) outliers in the first two (resp. last four) clusters.
  \item[Case 5 ($q\leq 6$):] $\bs \Sigma_5=\bo I_p$ with outliers generated in up to six directions via scale shifts with a covariance matrix $\bs {\tilde \Sigma}_5 = \diag(5,\ldots,5)$ if $p \leq 6$ and $\diag(5,5,5,5,5,5,1,\ldots,1)$ if $p > 6$.
  The 20 outliers are generated by drawing observations from a $N(\bo 0, \bs {\tilde \Sigma}_5)$ distribution and keeping the ones with at least one variable (among the first six) larger than the maximum value or smaller than the minimum value of the non-outlying observations.
\end{description}

Details concerning the implementation of the simulations can be found in the supplementary material.
To compare the performance of the methods, we provide the percentage of outliers correctly identified (denoted by TP for ``True Positive'')  and the percentage of non-outlying observations erroneously identified as outliers (FN for ``False Negative'').

\subsection{Selecting the invariant components}

Before examining the performance of ICS in terms of TP and FP, we observe the selected dimensions using the D'Agostino (DA) and the Parallel Analysis (PA) methods for a level $\alpha=5\%$.
Table~\ref{tab:ICSnoDir} below gives the average of these dimensions over the 1000 simulations for the different cases.
Note that the results for the other normality tests proposed in Subsection \ref{subsec:ICSselect} have not been reported because they do not improve the performance compared with the DA and PA methods.

\begin{table}[!h]
\centering
\begingroup\footnotesize
\begin{tabular}{lc*{6}{*{2}{c}}}
  \hline
Scatters & $p$ &  \multicolumn{2}{c}{Case 0} & \multicolumn{2}{c}{Case 1} &\multicolumn{2}{c}{Case 2} &\multicolumn{2}{c}{Case 3} & \multicolumn{2}{c}{Case 4} &\multicolumn{2}{c}{Case 5} \\
 &  & \multicolumn{2}{c}{($q=0$)}&\multicolumn{2}{c}{ ($q=1$)}& \multicolumn{2}{c}{($q=1$)}&\multicolumn{2}{c}{ ($q=2$)}& \multicolumn{2}{c}{($q=6$)}&\multicolumn{2}{c}{ ($q\leq 6$)}\\
 &  &  DA & PA &DA & PA&DA & PA&DA & PA&DA & PA&DA & PA\\
\hline
  \hline
COV - COV$_4$ &   6 & 0.14 & 0.08 & 1.06 & 1.58 & 1.00 & 1.00 & 1.96 & 2.90 & 2.67 & 6.00 & 1.34 & 5.96 \\
  COV - COV$_4$ &  25 & 0.42 & 0.09 & 1.25 & 1.98 & 1.27 & 1.09 & 2.09 & 4.33 & 2.95 & 10.48 & 1.62 & 8.41 \\
  COV - COV$_4$ &  50 & 0.80 & 0.06 & 1.59 & 1.82 & 1.53 & 2.02 & 2.37 & 4.35 & 2.93 & 11.48 & 1.99 & 7.41 \\
   \hline MLC - COV &   6 & 0.12 & 0.08 & 1.05 & 1.45 & 0.99 & 1.08 & 1.98 & 2.77 & 2.13 & 5.97 & 1.09 & 5.36 \\
  MLC - COV &  25 & 0.23 & 0.08 & 1.15 & 1.75 & 1.10 & 1.04 & 2.03 & 3.59 & 2.08 & 9.25 & 1.18 & 6.27 \\
  MLC - COV &  50 & 0.46 & 0.06 & 1.34 & 1.76 & 0.48 & 20.31 & 2.16 & 3.87 & 2.10 & 8.86 & 1.32 & 5.23 \\
   \hline MCD - COV &   6 & 0.15 & 0.05 & 1.06 & 1.05 & 1.00 & 1.00 & 2.01 & 2.21 & 2.06 & 6.00 & 1.08 & 5.62 \\
  MCD - COV &  25 & 0.38 & 0.07 & 1.28 & 1.29 & 1.21 & 1.02 & 2.15 & 2.84 & 2.08 & 9.24 & 1.13 & 6.56 \\
  MCD - COV &  50 & 0.65 & 0.05 & 1.46 & 1.51 & 1.43 & 1.06 & 2.33 & 3.33 & 1.79 & 6.94 & 1.13 & 3.45 \\
  \hline  MCD - MLC &   6 & 0.08 & 0.07 & 1.04 & 0.99 & 1.00 & 1.00 & 2.05 & 0.52 & 1.75 & 0.03 & 0.66 & 0.05 \\
  MCD - MLC &  25 & 0.25 & 0.05 & 1.17 & 1.14 & 1.11 & 1.04 & 2.09 & 2.40 & 1.28 & 1.96 & 0.68 & 1.54 \\
  MCD - MLC &  50 & 0.56 & 0.05 & 1.42 & 1.49 & 1.40 & 1.00 & 2.27 & 2.83 & 1.28 & 1.96 & 0.78 & 1.13 \\
   \hline
\end{tabular}
\caption{Averaged numbers of selected invariant components for the DA and PA methods}
 \label{tab:ICSnoDir}
\endgroup
\end{table}

Under setups 0, 1, 2 and 3, the results from Table~\ref{tab:ICSnoDir} are overall quite good and comparable for the different scatter pairs.
Only certain specific results have to be pointed out for the pairs $\mlt-\cov$ (Case 2 with $p=50$ for DA and PA) and $\mcd-\mlt$ (Case 3, $p=6$ for PA), and these points
require further investigation.
Moreover, for the four setups,  the differences between procedures DA and PA are small, with some overestimation of the dimension for PA in Case 3 when
$p=25$ or 50.

The results are not as good for Cases 4 and 5 which correspond to larger $q$ values than  the other setups, in particular for the DA procedure that leads
to an important underestimation of the dimension for all scatter pairs. The PA procedure gives better results in this context except for the $\mcd-\mlt$ pair, which leads
to an important underestimation in all cases.
The results for $\cov-\cov4$ and $\mcd-\cov$ are quite similar despite a larger overestimation of the dimension for $\cov-\cov4$ in Cases 4 and 5, for $p=25$ and $p=50$, when using the PA procedure.

These first results are in favor of the pairs $\cov-\cov4$ and $\mcd-\cov$ but need to be confirmed by studying the performance of the methods in terms of TP and FP.

\subsection{Detecting outliers with ICS}

Table~\ref{tab:ICSglobal} gives the TP (except for Case 0) and FP averaged over the 1000 simulations and averaged also over Cases 1 to 5 to save space.
The $\gamma$ level for the identification cut-off is fixed at 2\%.

\begin{table}[!h]
\centering
\begingroup
\footnotesize
 \begin{tabular}{l*{3}{*{3}{c}}}

  \hline
Averaged Measures in \%& \multicolumn{3}{c}{TP}& \multicolumn{3}{c}{FP} & \multicolumn{3}{c}{FP Case 0} \\
$p$ & 6 & 25 & 50 & 6 & 25 & 50&  6 &  25 & 50 \\
	  \hline
  		  \hline
				
  True subspace &  95.10 & 96.68 & 93.92 & 0.10 & 0.07 & 0.12 &  &  &  \\
  \hline
	ICS true $q$ COV - COV$_4$ & 96.92 & 92.52 & 80.13 & 0.57 & 0.33 & 0.49 &  &  &  \\
  ICS true $q$ MLC - COV & 97.53 & 92.24 & 64.97 & 1.27 & 0.63 & 1.01 &  &  &  \\
  ICS true $q$ MCD - COV & 97.53 & 93.59 & 81.99 & 1.48 & 0.92 & 0.83 &  &  &  \\
  ICS true $q$ MCD - MLC & 97.29 & 92.75 & 80.09 & 1.82 & 1.49 & 1.07 &  &  &  \\
   \hline
	ICS DA COV - COV$_4$ & 77.01 & 78.14 & 70.26 & 0.43 & 0.38 & 0.57 & 0.30 & 0.70 & 1.17 \\
	ICS DA MLC - COV & 76.34 & 75.84 & 54.18 & 1.03 & 0.56 & 0.65 & 0.25 & 0.42 & 0.76 \\
	ICS DA MCD - COV & 76.77 & 77.61 & 69.31 & 1.22 & 0.94 & 0.86 & 0.30 & 0.68 & 0.95 \\
	ICS DA MCD - MLC & 71.51 & 71.41 & 65.18 & 1.53 & 1.23 & 0.95 & 0.18 & 0.48 & 0.87 \\
 		\hline
	ICS PA COV - COV$_4$ & 96.73 & 91.39 & 76.29 & 0.70 & 0.64 & 0.79 & 0.10 & 0.11 & 0.09 \\
  ICS PA MLC - COV & 96.89 & 91.68 & 62.77 & 1.39 & 0.85 & 1.10 & 0.15 & 0.12 & 0.09 \\
  ICS PA MCD - COV & 97.36 & 93.36 & 76.66 & 1.50 & 1.03 & 0.89 & 0.11 & 0.14 & 0.09 \\
  ICS PA MCD - MLC & 44.95 & 76.92 & 65.44 & 0.95 & 1.32 & 0.90 & 0.08 & 0.11 & 0.09 \\
 	 		  \hline
\end{tabular}
\caption{TP and FP results for ICS (averaged results for Cases 1 to 5).}\label{tab:global_res}
\label{tab:ICSglobal}
\endgroup
\end{table}

The first row of Table~\ref{tab:ICSglobal} gives some kind of oracle performance measure obtained by calculating TP and FP values using the Euclidian norm of the projected data on the true subspace containing the outliers (known for each Case). As anticipated, the results are very good regardless of the dimension.
The next results are obtained when the true number of invariant components is selected but the invariant components are estimated using different scatter pairs.
They give another oracle performance measure.
Compared with the first row of Table~\ref{tab:ICSglobal}, these results are globally good in terms of TP but perform less well in terms of FP.
They give an idea of the impact of the scatter matrices estimation when the number of invariant components is the true one.
In this context, the COV - COV4 scatter pair clearly outperforms the others with similar TP values but smaller FP values.

Then, the results are given for the two automated selection DA and PA. Compared with the previous results, they give some insight into the impact of the dimension selection procedures. When looking at Cases 1 to 5, there is no method for dimension selection that outperforms the other.
PA is the best in terms of TP, but D'Agostino is the best in terms of FP.
However, for Case 0, any dimension $p$ and any scatter pair, the PA selection leads to less than 0.15\% FP on average, while
the values are larger for DA. The choice $\cov-\cov_4$ is the best in most situations from the FP point of view.

\begin{figure}
  \caption{Averaged TP and FP results for ICS detailed for Cases 1 to 5.}\label{fig:ICS_TP_FP_CASE15}
  \includegraphics[width=0.99\textwidth]{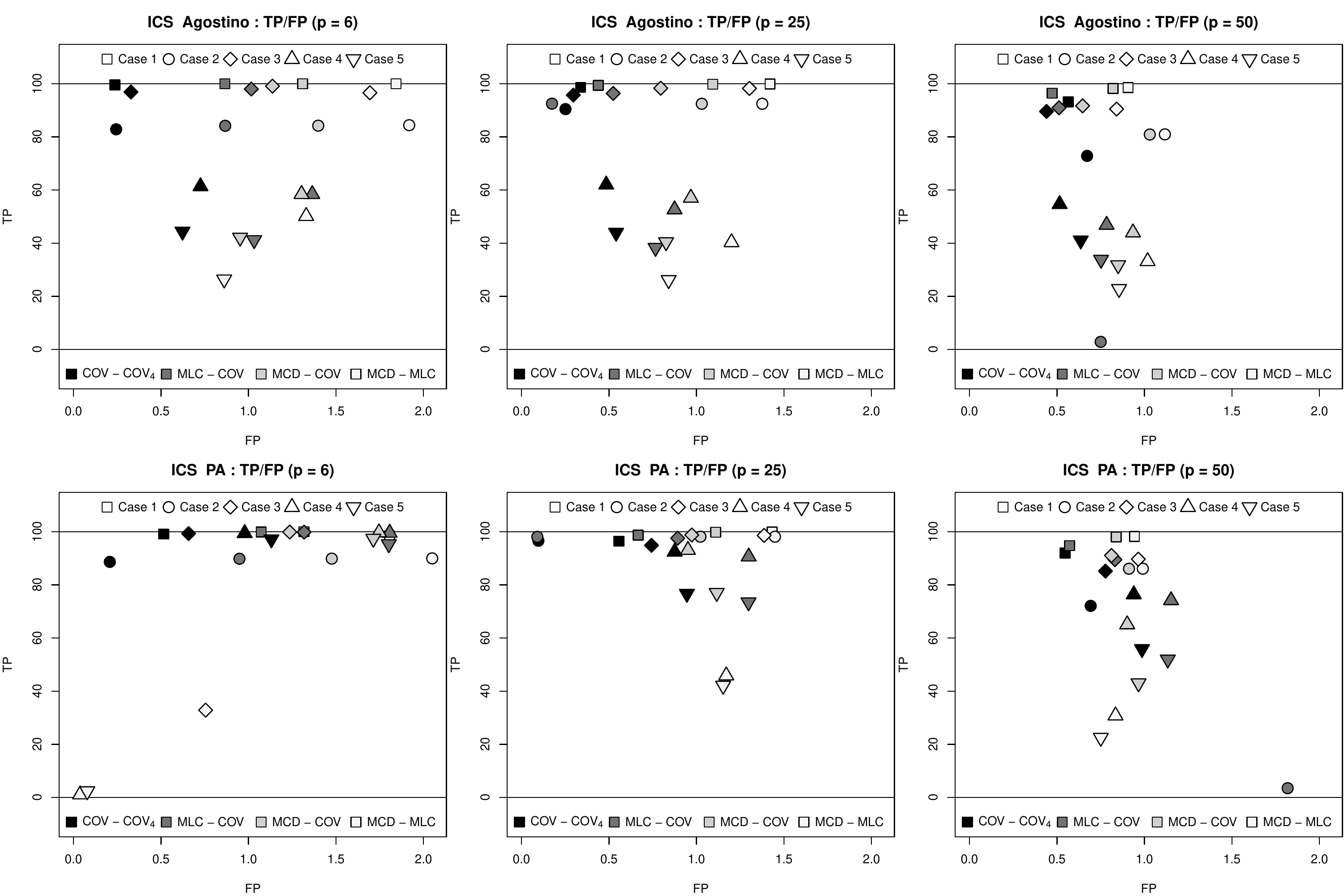}\\
\end{figure}

Figure~\ref{fig:ICS_TP_FP_CASE15} gives more details concerning the TP and the FP values for the Cases 1 to 5.
It contains scatter plots of the TP against the FP for D'Agostino (top) and PA (bottom) and for the different values of $p$.
Note that Case 2 for $p=50$ and $\mlt-\cov$  is very specific with often no component selected and will not be considered further in our comments.
For DA, the results are clearly ordered in terms of TP according to the different Cases, from the largest TP values for Case 1 to the smallest ones
for Case 5. There are only tiny differences between the scatter pairs.  With respect to FP values, the results are now ordered according to the different scatter pairs from
the smallest values for $\cov-\cov_4$ to the largest values for $\mcd-\mlt$. These differences are more limited for Cases 4 and 5 than for Cases 1 to 3 and decrease for all cases when $p$ increases.

For PA, the results differ. If we except $\mcd-\mlt$, all scatter pairs lead to very similar and good TP values when $p=6$ while $\cov-\cov_4$ is clearly the best when comparing FP values. For the particular pair $\mcd - \mlt$ and Cases 4 and 5, as observed from Table~\ref{tab:ICSnoDir}, no dimension is selected,
and so no outlier can be detected. When $p$ increases, in general the results become worse for TP, in particular for Cases 4 and 5, while they become close together for FP.

From this simulation results, we recommend using the pair $\cov-\cov4$. For this scatter pair, the results for DA and PA, - compared to the ones obtained when the
true dimension $q$ is known, - do not make it possible to conclude in favor of one of the two selection methods. While the TP values are better and closer to the oracle  for PA, the FP values are better and closer to the oracle  for DA.

\subsection{Comparing ICS with the Mahalanobis distance and PCA}
Table~\ref{tab:ICS_PCA_MD}  recalls the TP and the FP values for ICS focusing on $\cov-\cov4$ but also gives the values
when using non-robust (MD) and robust (RD) Mahalanobis distances  and PCA (unstandardized and standardized). RD is obtained using
a 25\% breakdown point reweighted MCD estimator.
For the Mahalanobis distances, we only report the results when the cut-off values are the usual ones, based on a chi-squared distribution quantile (of order 2\%) or are adjusted to take into account some asymptotic corrections for RD and the method is denoted GM (see \cite{packagecerioli}  and \cite{green2017}for the implementation).
Other criteria obtained through simulations have been implemented but do not bring any improvement and are not reported.
Concerning PCA and robust PCA, the outlier detection procedure is quite complex since the method is not aimed at detecting outliers.
Atypical observations may thus be displayed on any of the $p$ principal components \citep{jolliffe2002}.
Basically, the procedure consists in selecting some components and calculating, on the one hand, a distance in the space spanned by the selected components (after some standardization), and, on the other hand, a distance in the space orthogonal to the previous space (see \cite{hubert2005robpca} for details). In our comparison, following \citet{hubert2005robpca}, observations associated with at least one large distance are flagged as outliers using some cut-off values based on quantiles of order 99\% for each distance. We tried different methods for principal components selection but report only the results obtained when the dimension is chosen as the best possible among all possible dimensions (from 1 to $p$). More precisely, it means that the results give the smallest FP value among all the results that were found to maximize the TP value. Automated methods were also tested but the results were never better than the ones reported. Some robust PCA methods where the usual covariance or correlation matrix is replaced by some robust estimators were also implemented but did not lead to better results and are not reported neither.
Results are averaged for Cases 1 to 5 in order to save space.

\begin{table}[!h]
\centering
\begingroup
 \begin{tabular}{l*{3}{*{3}{c}}}
  \hline
Averaged Measures & \multicolumn{3}{c}{TP}& \multicolumn{3}{c}{FP} & \multicolumn{3}{c}{FP Case 0} \\
$p$ & 6 & 25 & 50 & 6 & 25 & 50&  6 &  25 & 50 \\
	  \hline
		 \hline
  ICS DA COV - COV$_4$ & 77.01 & 78.14 & 70.26 & 0.43 & 0.38 & 0.57 & 0.30 & 0.70 & 1.17 \\
 % ICS DA MLC - COV & 76.34 & 75.84 & 54.18 & 1.03 & 0.56 & 0.65 & 0.25 & 0.42 & 0.76 \\
 % ICS DA MCD - COV & 76.77 & 77.61 & 69.31 & 1.22 & 0.94 & 0.86 & 0.30 & 0.68 & 0.95 \\
 % ICS DA MCD - MLC & 71.51 & 71.41 & 65.18 & 1.53 & 1.23 & 0.95 & 0.18 & 0.48 & 0.87 \\
  ICS PA COV - COV$_4$ & 96.73 & 91.39 & 76.29 & 0.70 & 0.64 & 0.79 & 0.10 & 0.11 & 0.09 \\
 % ICS PA MLC - COV & 96.89 & 91.68 & 62.77 & 1.39 & 0.85 & 1.10 & 0.15 & 0.12 & 0.09 \\
 % ICS PA MCD - COV & 97.36 & 93.36 & 76.66 & 1.50 & 1.03 & 0.89 & 0.11 & 0.14 & 0.09 \\
 % ICS PA MCD - MLC & 44.95 & 76.92 & 65.44 & 0.95 & 1.32 & 0.90 & 0.08 & 0.11 & 0.09 \\
   \hline
		 \hline
	MD & 94.14 & 72.80 & 52.03 & 1.24 & 1.91 & 2.40 & 2.09 & 2.35 & 2.69 \\
 % MD thq & 93.98 & 70.50 & 46.63 & 1.17 & 1.51 & 1.52 & 1.98 & 1.87 & 1.71 \\
   \hline
 %	RD COV$_4$ & 63.45 & 53.02 & 39.27 & 0.86 & 1.39 & 1.31 & 2.07 & 1.83 & 1.54 \\
 % RD thq COV$_4$ & 62.83 & 51.65 & 36.72 & 0.82 & 1.23 & 1.04 & 1.96 & 1.62 & 1.21 \\
 % RD MLC & 97.29 & 86.71 & 52.23 & 1.85 & 1.82 & 1.99 & 2.10 & 2.09 & 2.09 \\
 % RD thq MLC & 98.04 & 91.94 & 59.96 & 8.34 & 4.12 & 3.77 & 9.23 & 4.64 & 3.96 \\
	RD GM & 94.03 & 78.80 & 53.57 & 0.20 & 0.26 & 0.26 & 0.21 & 0.29 & 0.23 \\
  RD & 97.37 & 91.34 & 75.26 & 1.78 & 1.88 & 1.94 & 2.09 & 2.12 & 2.10 \\
 % RD thq MCD & 97.38 & 92.22 & 81.02 & 1.80 & 2.28 & 3.53 & 2.11 & 2.55 & 3.80 \\
	 \hline
		 \hline
% PCA COV$_4$ & 83.96 & 91.65 & 85.73 & 1.02 & 1.11 & 1.07 & 2.00 & 1.78 & 1.53 \\
  PCA & 98.55 & 91.58 & 84.43 & 1.14 & 1.22 & 1.17 & 2.01 & 1.91 & 1.84 \\
% PCA MLC & 99.89 & 95.59 & 88.52 & 5.96 & 3.49 & 3.42 & 3.41 & 2.11 & 1.82 \\
%  PCA MCD & 99.88 & 91.02 & 59.37 & 6.02 & 3.52 & 3.15 & 3.37 & 2.10 & 1.81 \\
%  k0.PCA COV$_4$ & 78.98 & 43.03 & 28.37 & 1.40 & 1.65 & 1.67 & 2.38 & 1.87 & 1.81 \\
%  k0.PCA & 85.45 & 49.11 & 32.50 & 1.53 & 1.69 & 1.70 & 2.48 & 1.91 & 1.84 \\
%  k0.PCA MLC & 86.37 & 55.42 & 37.81 & 5.71 & 2.81 & 2.40 & 6.38 & 3.01 & 2.48 \\
%  k0.PCA MCD & 99.58 & 84.63 & 48.19 & 5.84 & 2.68 & 2.23 & 6.38 & 2.99 & 2.48 \\
%  k90.PCA COV$_4$ & 58.57 & 58.06 & 33.24 & 1.53 & 1.57 & 1.46 & 2.35 & 1.78 & 1.61 \\
%  k90.PCA & 98.16 & 62.68 & 39.74 & 1.18 & 1.71 & 1.67 & 2.43 & 1.97 & 1.84 \\
%  k90.PCA MLC & 99.88 & 70.77 & 48.09 & 5.99 & 3.23 & 2.81 & 6.74 & 3.36 & 2.88 \\
%  k90.PCA MCD & 99.87 & 87.93 & 50.52 & 5.99 & 3.00 & 2.59 & 6.72 & 3.36 & 2.85 \\
%  PCA std COV$_4$ & 99.28 & 94.36 & 89.56 & 1.95 & 2.32 & 3.18 & 1.95 & 2.04 & 2.03 \\
  PCA std & 80.98 & 80.88 & 47.85 & 0.58 & 0.84 & 1.52 & 1.99 & 1.80 & 1.54 \\
%  PCA std MLC & 95.27 & 69.69 & 46.81 & 2.00 & 2.08 & 1.79 & 1.97 & 1.89 & 1.80 \\
%  PCA std MCD & 99.25 & 93.53 & 85.78 & 1.96 & 2.33 & 3.09 & 1.93 & 2.05 & 2.03 \\
%  k0.PCA std COV$_4$ & 85.50 & 59.65 & 49.79 & 1.92 & 2.02 & 3.08 & 2.34 & 2.29 & 3.48 \\
%  k0.PCA std & 59.07 & 50.08 & 37.01 & 1.67 & 1.50 & 1.51 & 2.38 & 1.87 & 1.82 \\
%  k0.PCA std MLC & 93.83 & 64.39 & 41.08 & 1.59 & 1.60 & 1.65 & 2.19 & 1.90 & 1.87 \\
%  k0.PCA std MCD & 99.05 & 89.73 & 78.82 & 1.89 & 1.94 & 3.05 & 2.34 & 2.31 & 3.48 \\
%  k90.PCA std COV$_4$ & 99.04 & 74.12 & 57.35 & 2.19 & 2.20 & 2.84 & 2.74 & 2.30 & 3.21 \\
%  k90.PCA std & 61.04 & 46.64 & 33.40 & 1.48 & 1.51 & 1.40 & 2.35 & 1.80 & 1.61 \\
%  k90.PCA std MLC & 94.42 & 64.58 & 40.63 & 1.96 & 1.62 & 1.66 & 2.65 & 1.91 & 1.86 \\
%  k90.PCA std MCD & 98.97 & 91.53 & 81.00 & 2.27 & 1.97 & 2.79 & 2.78 & 2.31 & 3.22 \\
  \hline
\end{tabular}
\caption{Comparison of ICS with MD, RD and PCA (averaged results in \% for Cases 1 to 5).}
 \label{tab:ICS_PCA_MD}
\endgroup
\end{table}

The performance of MD, RD and PCA compared to the other methods is particularly low when focusing on the FP measure.
For standardized PCA,
 results are better when the dimension is equal to 6 but the method cannot compete when the dimension increases.
ICS with DA and PA together with RD when using the GM correction lead to better performance. When $p=6$, RD GM gives the best results with very low FP on average for Cases 1 to 5. When $p=50$, the method still leads to low false detection but at the cost of a low true positive detection compared to ICS. For Case 0, RD GM exhibits good results but ICS PA outperforms it. In conclusion,  we advocate the use of ICS with the scatter pair $\cov-\cov4$, which is very easy to compute and exhibits good performance. In this framework where the majority of the data follows a Gaussian distribution, we recommend the PA components selection method, but the DA method is an interesting alternative with a very low computational cost. Note that in case the majority of the data does not follow a Gaussian distribution, the different cutoffs are not valid anymore and should be adapted.

\section{Data Analysis}
\label{sec_real}
We analyze three real data sets using ICS and compare ICS with several competitors that are Mahalanobis distance or PCA variants, including ROBPCA as introduced by \citet{hubert2005robpca} and implemented in the package rrcov \citep{rrcov}.
All data are from industrial processes and contain potentially a small proportion of outliers. The last two data sets in particular come from industrial processes where there is a high level of quality control and only a small proportion of observations can be diagnosed as outliers. It implies that the False Positive rate is crucial and should be as small as possible.

For each of the three data sets, we give details concerning the observations considered as true outliers in the following subsections. Table \ref{tab:realex} provides the number of True Positive (NTP) and the number of False Positive (NFP) for the three data sets. Be careful that for this particular table, the values are not given as proportions. For ICS, we only report results for the scatter pair $\cov - \cov_4$ because, in most cases, this pair leads  to the best results, which is consistent with our simulation conclusions and confirms our recommendations. For ICS and PCA methods, the results depend on the number of selected components, and we show results for three different types of selection. The ``best selection'' results are obtained by trying all possible dimensions between 1 and $p$ and taking, for each method, the dimension $k$, which leads to the smallest NFP among those that maximize the NTP. This procedure leads to some kind of oracle measure of the maximum performance of the methods.
The second type of results are obtained through automated components selection methods as detailed in the previous section for ICS and using the rule proposed in the package rrcov for ROBPCA.
Moreover, for ICS, only the DA and PA automated components selection methods are reported,
because they give the best results in general.
As can be observed from the last two data sets, and also from our experience on other data sets, the automated procedures for ICS tend to select too many components.
One possible reason is that these procedures rely on the Gaussian distribution of the main bulk of the data, and such an assumption may not be fulfilled in practice.
Therefore, we propose to use the scree plot as an alternative visual selection method  that leads to a third type of results for ICS. The scree plot is very well-known for PCA \citep{jolliffe2002} and can be applied in the same way for ICS except that for the scatter pair $\cov-\cov4$, the eigenvalues are to be interpreted in terms of kurtosis (see \citet{Tyler2009}), instead of variance for PCA. The scree plots for the three examples are given on Figure \ref{fig:scree}. For the three scree plots, some invariant components (two for the Glass and the Reliability data sets and three for the HTP data) clearly differ from the other components due to their high eigenvalues. The results for these components selection are reported in the last row of Table \ref{tab:realex}.

\begin{figure}
  \center
  \includegraphics[width=0.7\textwidth]{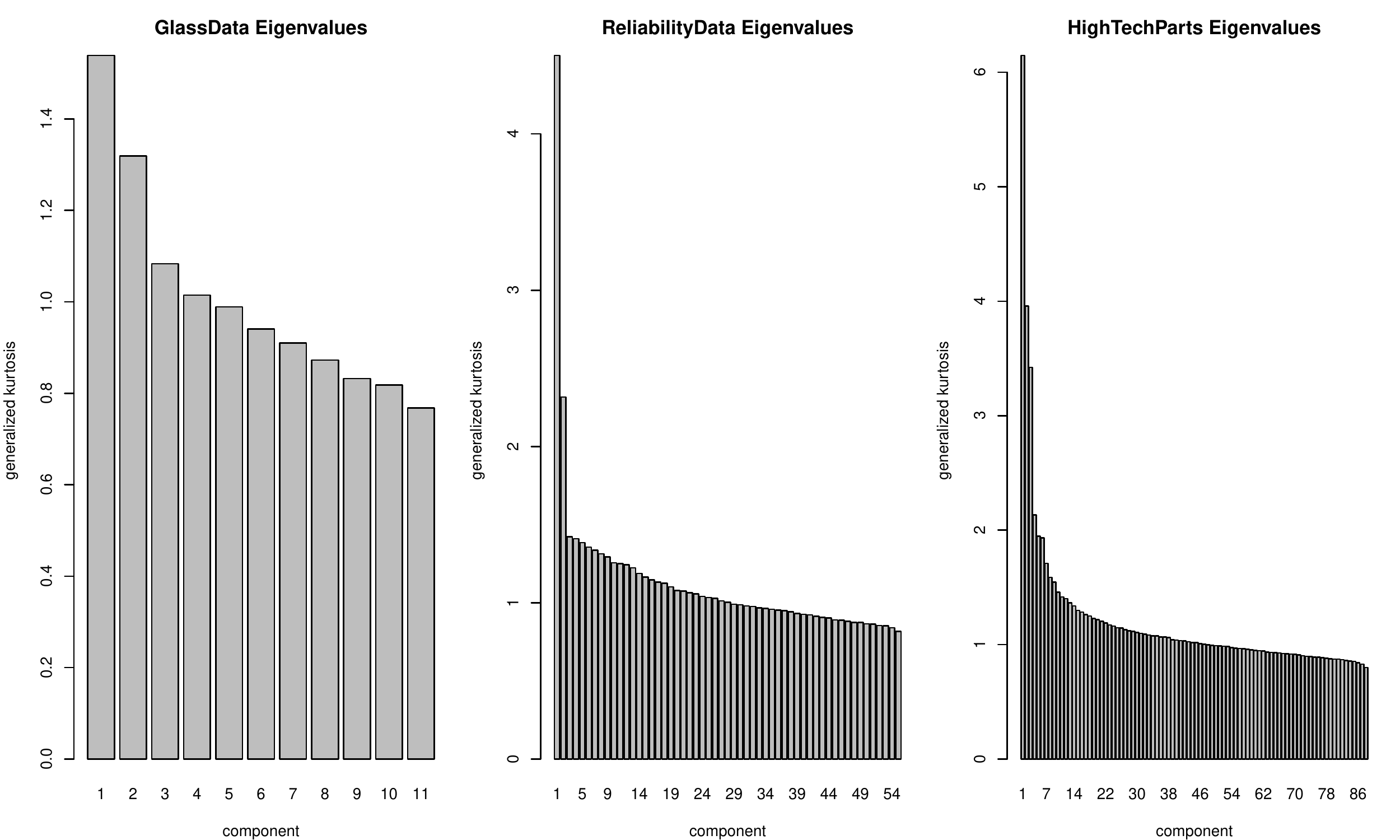}\\
	\caption{Scree plots for ICS with $\cov-\cov4$ for the three data sets.}\label{fig:scree}
\end{figure}

\begin{table}[!h]
\centering
\begingroup\scriptsize
 \begin{tabular}{l*{3}{ccc}}

  \hline
 & \multicolumn{3}{c}{Glass}& \multicolumn{3}{c}{Reliability} & \multicolumn{3}{c}{HighTech} \\
 &  NTP (/3) & NFP (/109) & $k$ (/11) & NTP (/2) & NFP (/518) & $k$ (/55) & NTP (/2) & NFP (/900) & $k$ (/88) \\

  \hline
	MD & 3 & 4 &  & 2 & 52 &  & 2 & 119 &  \\
  RD & 3 & 15 &  &  &  &  & 2 & 243 &  \\
	RD GM & 3 & 7 &  &  &  &  & 2 & 223 &  \\
	\hline
	{\em Best selection} &&&&&&&&&\\
  ICS $\cov - \cov_4$  & 3 & 3 & 2 & 2 & 1 & 1 & 2 & 0 & 1 \\
  PCA & 3 & 9 & 5 & 2 & 41 & 52 & 2 & 21 & 1 \\
  PCA std & 3 & 4 & 2 & 2 & 22 & 40 & 2 & 25 & 6 \\
  ROBPCA & 3 & 13 & 5 &  &  &  & 2 & 50 & 1 \\
	\hline
 {\em Automated selection} &&&&&&&&&\\
  ICS $\cov - \cov_4$ DA & 3 & 3 & 2 & 2 & 23 & 12 & 2 & 39 & 14 \\
  ICS $\cov - \cov_4$ PA & 3 & 3 & 2 & 2 & 42 & 28 & 2 & 87 & 50 \\
  PCA & 1 & 5 & 1 & 0 & 6 & 12 & 2 & 24 & 3 \\
  PCA std & 1 & 4 & 1 & 2 & 31 & 20 & 2 & 28 & 4 \\
  ROBPCA & 3 & 17 & 1 &  &  &  & 2 & 80 & 2 \\
	\hline
  {\em Scree plot selection} &&&&&&&&&\\
  ICS  $\cov - \cov_4$ & 3 & 3 & 2 & 2 & 1 & 2 & 2 & 5 & 3 \\
  \hline
\end{tabular}
\caption{NTP, NFP and number $k$ of selected components for the three real data examples.}
\label{tab:realex}
\endgroup
\end{table}

\subsection{Glass recycling}
The so-called glass data set is analyzed by \citet{cerioli2011error} and consists of  112 glass fragments collected for recycling, of which 109 are true glass fragments and 3 are contaminated ceramic glass fragments. The 11 variables are the log of spectral measures recorded for each fragment.
For all methods, the outliers are flagged by using cut-offs defined through simulations at the 5\% level so that results are comparable with  \citet{cerioli2011error}.
For this example, ICS detects the three outliers and has only three false detections, and the results are the same for the three types of components selection (best, automated or scree plot).
ICS has the highest performance in comparison with the competitors considered here  but also in comparison with the results reported in Table 6 of \citet{cerioli2011error}.
The non robust Mahalanobis distance, which is equivalent to ICS with $\cov - \cov_4$ when all components are selected, also performs quite well on this example. All three outliers are detected, and there are only four false detections compared to three when two invariant components are selected among the eleven.
%
%
%
%\begin{figure}
%  \center
%  % Requires \usepackage{graphicx}
%\includegraphics[width=\textwidth]{GlassData_MD_ICS_11_variables}
%  \caption{}\label{fig2:glass}
%\end{figure}

For the next two examples, to obtain an acceptable quality control performance, true outliers should be detected with  up to 2\% of observations flagged as outliers, taking into account the true outliers and the false detections.
Moreover, the results for these two examples are readily reproducible using the R code provided in the supplementary material of the present paper.

\subsection{Reliability Data}
The Reliability data are available in the R package REPPlab \citep{REPPlab} and contain 55 variables measured on 520 units during a production process. The quality standards for this process are respected for each variable, and the objective is to detect some potential multivariate faulty units representing less than 2\% of the 520 observations. In \citet{fischer2016}, two observations (414 and 512) are detected as the most severe outliers. For simplicity, we consider these two observations as the only true outliers. However, there may be other outliers, and the NFP numbers should be viewed with caution for this example in comparison with the other two data sets, where some auxiliary information concerning the true outliers is known.
For this example and the next one, the outliers are flagged by using cut-offs defined through simulations at the 2\% level.

In Table \ref{tab:realex}, the results for the MCD are not reported. As mentioned in \citet{fischer2016}, computing the MCD (at least with a breakdown point equal or larger than 25\%) is not possible on this data set because 497 observations among the 520 take exactly the same value on the 24th variable. Note that from our experience, this problem occurs quite
recurrently on real data sets in some industrial context, and, as illustrated below, removing such variables may lead to a loss of relevant information.
This is, however, not a problem for ICS when using the scatter pair $\cov - \cov_4$, and the method shows very good performance for the Reliability data when selecting only two components.
The only observation declared as a false positive is observation 57, which is also flagged as an outlier in \citet{fischer2016} (although not as extreme as the other two).
The selection of  two components is suggested by the scree plot analysis.
The automated selection procedures or the use of all invariant components (Mahalanobis distance) show poor performance with a number of false positives higher than the 2\% rate that is acceptable. PCA is even less successful with many false positive in the best selection case and sometimes no detection at all when the components selection is automated.

Moreover, when the number of selected invariant components is small, ICS makes the detected outliers easy to interpret by drawing scatter plots of the selected components and by observing the correlations between the components and the original variables.
Figure \ref{fig2:reliability} illustrates this point. The two selected invariant components are plotted on the left panel and clearly lead to the identification of observations 414 and 512 as outliers. When calculating the correlations between these invariant components and the 55 original variables, it appears that they are essentially correlated with variables 22 and 24. These two variables are thus plotted on the right panel of Figure \ref{fig2:reliability} and reveal that observation 414 (resp. 512) combines in an unusual way a high (resp. small) value on variable 22 with a small (resp. large) value on variable 24. Note that removing variable 24 in order to compute the MCD estimate precludes the ability to detect the two outliers.

\begin{figure}
  \center
  \includegraphics[width=\textwidth]{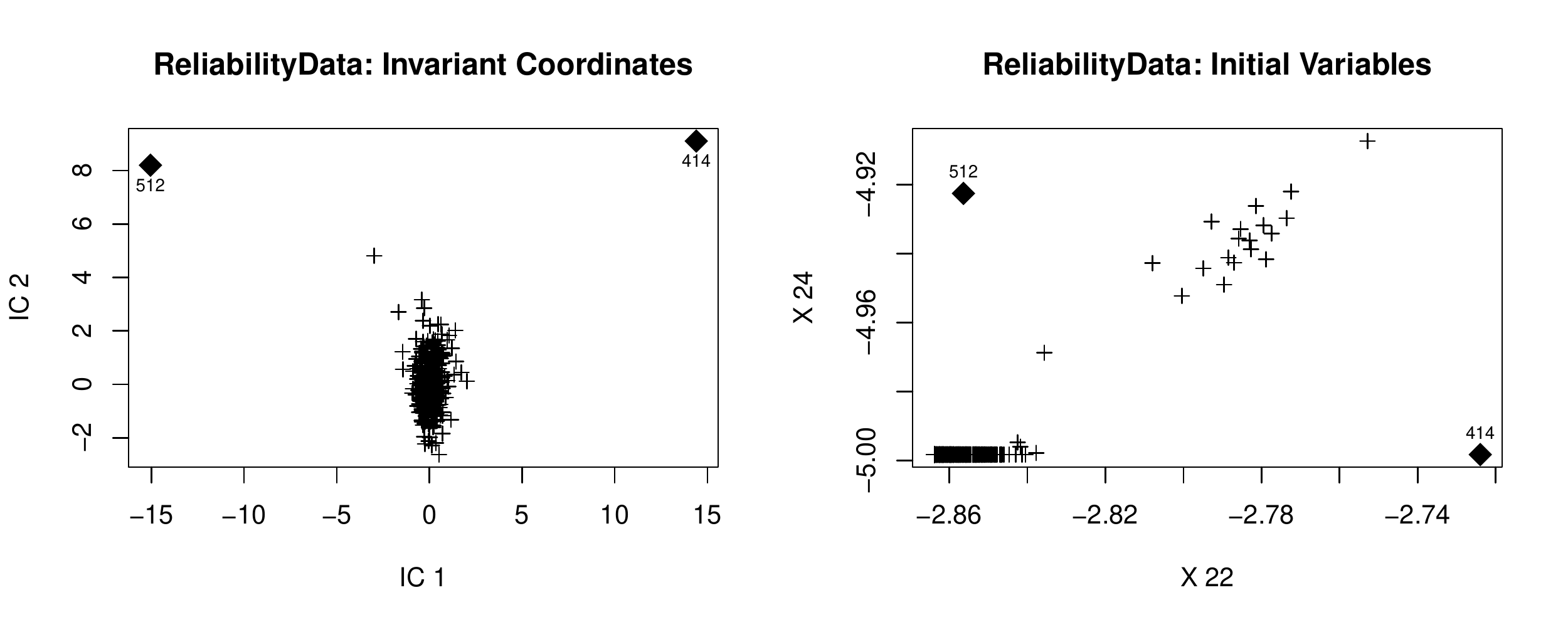}\\
  \caption{Scatter plot of the first two invariant components (left panel) and scatter plot of the variables numbered 22 and 24 (right panel) for the Reliability data set.}\label{fig2:reliability}
\end{figure}

\subsection{High-tech parts}
The third real data set contains 902 high-tech parts designed for consumer products and characterized by 88 electronic measures; it is available in the R package ICSOutlier \citep{ICSOutlier}. To anonymize the data collected, the measures have been mean-centered. We do not have access to the original data, but we know that they were cleaned from univariate outliers using some preliminary standard quality control rules. No multivariate outlier detection method was applied and the parts were sold. However, two parts (denoted by R1 and R2 in what follows) among the 902 were found to be defective and returned to the manufacturer. Our objective is to check whether these two observations could have been detected before being sold, using some multivariate outlier detection method in an unsupervised way, with less than 2\% of observations flagged as outliers.

From Table \ref{tab:realex}, the result based on only one component (best selection) for ICS is perfect, with two outliers detected and no false detection. The results are  much worse for all other methods, with too many false detections. This is especially true when considering the Mahalanobis distance with no selection of components. The results for ICS are rather mediocre when using the DA or (even worse) the PA automated selection methods which tend to select too many components. Using the scree plot, however, leads unambiguously to a more drastic selection, with three eigenvalues larger than the others. Using three components leads to good performance, with five NFP and all together seven detected outliers, which is less than 2\% and thus acceptable.
Figure \ref{fig:htp} gives more insight on the influence of the number of selected invariant components on the detection performance and echoes Figure \ref{fig:MD1}. The six scatter plots give the squared ICS distances when the number of components increases. The top-left plot corresponds to one component, which is the best possible selection. Then, the NFP increases when more components are selected. The bottom-left plot corresponds to DA selection, while the bottom-middle plots correspond to PA selection. On the bottom-right plot, all 88 components are taken into account, which corresponds to the squared Mahalanobis distance, and the result is the worst. Note that for this data set, PCA performs better than the Mahalanobis distance even if the number of NFP is still unacceptable.

\begin{figure}
  \center
  \includegraphics[width=\textwidth]{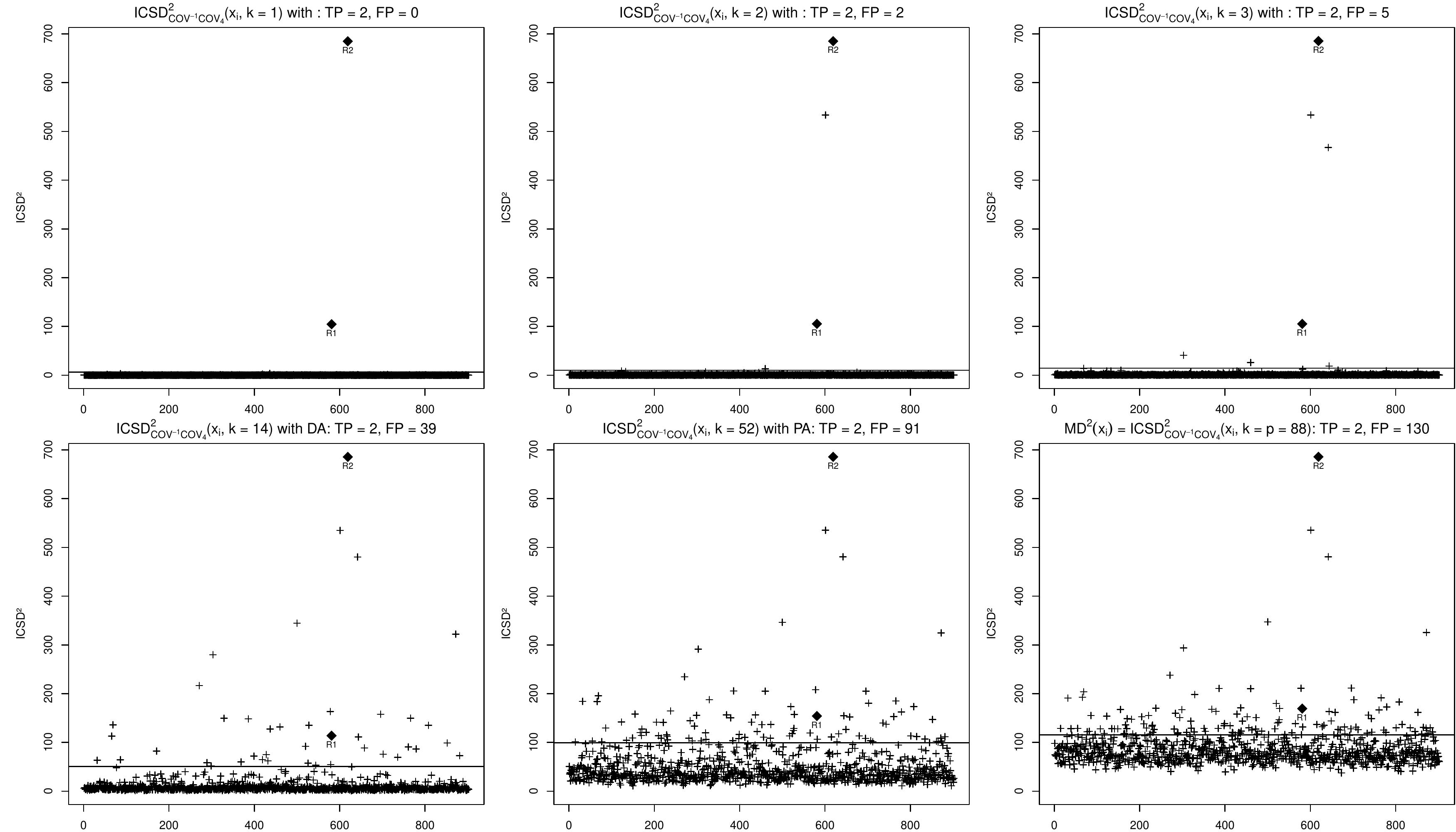}\\
  \caption{Plots of the squared ICS distances for different numbers of invariant components selected for the HTP data set.}\label{fig:htp}
\end{figure}

Finally, ICS is shown to be appropriate for the three data sets when using the scree plot selection method, while the performance of its competitors depends on the data set.

\section{Conclusion and perspectives}\label{sec_concl}

The remarkable theoretical properties of ICS are confirmed in the context of multivariate outlier detection with a small proportion of outliers.
In particular, the ability of ICS to recover the Fisher's linear discriminant subspace in the case where group identifications are unknown has been verified on simulations, with a majority of the data following a Gaussian distribution, but also on some real data set where the Gaussian assumption is not true anymore.
So, as stated for Linear Discriminant Analysis by \citet{hastie01statisticallearning}, it seems that the applicability of ICS extends beyond the realm of Gaussian data.
However, this remark does not apply to the components selection procedures we propose.
From our simulation study, we advocate the use of some selection methods such as DA or PA. But such methods are not convincing when analyzing real data sets as they tend to select too many components. The reason is certainly the fact that the majority of the data does not follow a Gaussian distribution while the cut-offs we propose depend heavily on this assumption.
The data analysis of real data sets highlights the advantage of using the scree plot for selecting the number of components.

Contrary to PCA, the method is not only orthogonal invariant but also scale invariant and is aimed at detecting outliers.
More precisely, the present paper demonstrates the good performance of ICS, when using the scatter pair $\cov - \cov_4$
and selecting the first components in a context of a small proportion of outliers.
The simulation study together with the data analysis illustrates that ICS consistently detects outliers, when they are present, with a small proportion of false detections, while the success of its competitors depends more on the data set under study. This is particularly true for PCA whose results depend a lot on the way the data are scaled. If the outliers are not concentrated in a small dimension subspace, ICS is equivalent to the Mahalanobis distance. For large dimensions and when outliers are contained in a small dimensional subspace, using ICS may improve greatly with respect to the Mahalanobis distance as illustrated by some theoretical properties and applications.
Moreover, selecting a small number of invariant components makes outlier interpretation much easier.

A perspective of the work is to consider multiple testing procedures for the choice of the cut-off for the distances as proposed by \citet{cerioli2010multivariate} and \citet{cerioli2011error}. Moreover, instead of defining cut-offs independently for the components selection and the outlier detection steps, it would be of interest to propose some alternative which would control the overall false positive rate of the global procedure.
Another perspective is to consider the case of a large proportion of outliers. In such a context, the scatter pair choice has to be revisited together with the components choice.
If outliers are contained in a small dimensional subspace, the $\cov - \cov_4$ pair, even if it is not robust, may still be a good alternative given the ICS theoretical properties. However, small kurtosis values are now also of interest, and thus invariant components associated with
small eigenvalues should be examined. In such a context, the recent papers \cite{nordhausen2016} and \cite{NordhausenOjaTylerVirta2017} are of particular interest.
The problem of high dimension and small sample size is also relevant in our industrial context for some particular applications.
The adaptation of ICS to such data sets is a work in progress.

%\end{linenumbers}

\section{Acknowledgements}
The work of Klaus Nordhausen was partly supported by the Academy of Finland (grant 268703).
%The authors wish to acknowledge Andrea Cerioli and Alesso Farcomeni for providing the glass data set.
The article is based upon work from CRoNoS COST Action IC1408, supported by COST (European Cooperation in Science and Technology).

\section*{References}
  \bibliographystyle{elsarticle-harv}
  \bibliography{Bibli5}

\begin{thebibliography}{48}
\expandafter\ifx\csname natexlab\endcsname\relax\def\natexlab#1{#1}\fi
\expandafter\ifx\csname url\endcsname\relax
  \def\url#1{\texttt{#1}}\fi
\expandafter\ifx\csname urlprefix\endcsname\relax\def\urlprefix{URL }\fi

\bibitem[{Aggarwal(2017)}]{aggarwal2017}
Aggarwal, C.~C., 2017. Outlier Analysis, 2nd edition. Springer Publishing
  Company, Incorporated.

\bibitem[{Agostinelli et~al.(2015)Agostinelli, Leung, Yohai, and
  Zamar}]{agostinelli2015robust}
Agostinelli, C., Leung, A., Yohai, V.~J., Zamar, R.~H., 2015. Robust estimation
  of multivariate location and scatter in the presence of cellwise and casewise
  contamination. Test 24~(3), 441--461.

\bibitem[{Alashwali and Kent(2016)}]{Alashwali2016}
Alashwali, F., Kent, J., 2016. The use of a common location measure in the
  invariant coordinate selection and projection pursuit. Journal of
  Multivariate Analysis 152, 145--161.

\bibitem[{Archimbaud et~al.(2016)Archimbaud, Nordhausen, and
  Ruiz-Gazen}]{ICSOutlier}
Archimbaud, A., Nordhausen, K., Ruiz-Gazen, A., 2016. ICSOutlier: Outlier
  Detection Using Invariant Coordinate Selection. R package version 0.2-0.
\newline\urlprefix\url{http://CRAN.R-project.org/package=ICSOutlier}

\bibitem[{Beckman and Cook(1983)}]{beckman1983outlier}
Beckman, R.~J., Cook, R.~D., 1983. Outliers. Technometrics 25~(2), 119--149.

\bibitem[{Bonett and Seier(2002)}]{bonett2002test}
Bonett, D.~G., Seier, E., 2002. A test of normality with high uniform power.
  Computational Statistics \& Data Analysis 40~(3), 435--445.

\bibitem[{Cator and Lopuha{\"a}(2012)}]{cator2012}
Cator, E.~A., Lopuha{\"a}, H.~P., 2012. Central limit theorem and influence
  function for the {MCD} estimators at general multivariate distributions.
  Bernoulli 18~(2), 520--551.

\bibitem[{Caussinus et~al.(2003)Caussinus, Hakam, and
  Ruiz-Gazen}]{caussinus2003projections}
Caussinus, H., Hakam, S., Ruiz-Gazen, A., 2003. Projections
  r{\'e}v{\'e}latrices contr{\^o}l{\'e}es: Groupements et structures diverses.
  Revue de Statistique Appliqu{\'e}e 51~(1), 37--58.

\bibitem[{Caussinus and Ruiz(1990)}]{caussinus1990interesting}
Caussinus, H., Ruiz, A., 1990. Interesting projections of multidimensional data
  by means of generalized principal component analyses. In: Compstat. Springer,
  pp. 121--126.

\bibitem[{Caussinus and Ruiz-Gazen(1994)}]{caussinus1994}
Caussinus, H., Ruiz-Gazen, A., 1994. Projection pursuit and generalized
  principal component analysis. New Directions in Statistical Data Analysis and
  Robustness, 35--46.

\bibitem[{Caussinus and Ruiz-Gazen(1995)}]{caussinus1995metrics}
Caussinus, H., Ruiz-Gazen, A., 1995. Metrics for finding typical structures by
  means of principal component analysis.

\bibitem[{Cerioli(2010)}]{cerioli2010multivariate}
Cerioli, A., 2010. Multivariate outlier detection with high-breakdown
  estimators. Journal of the American Statistical Association 105~(489),
  147--156.

\bibitem[{Cerioli and Farcomeni(2011)}]{cerioli2011error}
Cerioli, A., Farcomeni, A., 2011. Error rates for multivariate outlier
  detection. Computational Statistics \& Data Analysis 55~(1), 544--553.

\bibitem[{Cerioli et~al.(2009)Cerioli, Riani, and
  Atkinson}]{cerioli2009controlling}
Cerioli, A., Riani, M., Atkinson, A.~C., 2009. Controlling the size of
  multivariate outlier tests with the {MCD} estimator of scatter. Statistics
  and Computing 19~(3), 341--353.

\bibitem[{Croux et~al.(2013)Croux, Filzmoser, and Fritz}]{croux2013robust}
Croux, C., Filzmoser, P., Fritz, H., 2013. Robust sparse principal component
  analysis. Technometrics 55~(2), 202--214.

\bibitem[{Croux and Haesbroeck(1999)}]{croux1999}
Croux, C., Haesbroeck, G., 1999. Influence function and efficiency of the
  minimum covariance determinant scatter matrix estimator. Journal of
  Multivariate Analysis 71~(2), 161 -- 190.

\bibitem[{Dray(2008)}]{Dray2008}
Dray, S., 2008. On the number of principal components: A test of dimensionality
  based on measurements of similarity between matrices. Computational
  Statistics \& Data Analysis 52~(4), 2228 -- 2237.

\bibitem[{Fischer et~al.(2015)Fischer, Berro, Nordhausen, and
  Ruiz-Gazen}]{REPPlab}
Fischer, D., Berro, A., Nordhausen, K., Ruiz-Gazen, A., 2015. REPPlab: R
  Interface to EPP-Lab, a Java Program for Exploratory Projection Pursuit. R
  package version 0.9.2.
\newline\urlprefix\url{http://CRAN.R-project.org/package=REPPlab}

\bibitem[{Fischer et~al.(2016)Fischer, Berro, Nordhausen, and
  Ruiz-Gazen}]{fischer2016}
Fischer, D., Berro, A., Nordhausen, K., Ruiz-Gazen, A., 2016. {REPP}lab: An {R}
  package for detecting clusters and outliers using exploratory projection
  pursuit. Tech. rep., arXiv:1612.06518v1.

\bibitem[{Genz and Bretz(2009)}]{mvtnorm}
Genz, A., Bretz, F., 2009. Computation of Multivariate Normal and t
  Probabilities. Lecture Notes in Statistics. Springer-Verlag, Heidelberg.

\bibitem[{Green and Martin(2017{\natexlab{a}})}]{packagecerioli}
Green, C.~G., Martin, D., 2017{\natexlab{a}}. CerioliOutlierDetection: Outlier
  Detection Using the Iterated RMCD Method of Cerioli (2010). R package version
  1.1.9.
\newline\urlprefix\url{https://CRAN.R-project.org/package=CerioliOutlierDetection}

\bibitem[{Green and Martin(2017{\natexlab{b}})}]{green2017}
Green, C.~G., Martin, R.~D., 2017{\natexlab{b}}. An extension of a method of
  {Hardin} and {Rocke}, with an application to multivariate outlier detection
  via the {IRMCD} method of {Cerioli}. Tech. rep., Working Paper, 2017.
\newline\urlprefix\url{http://christopherggreen.github.io/papers/hr05\_extension.pdf}

\bibitem[{Greene(2012)}]{greene2012}
Greene, W., 2012. Econometric Analysis. Pearson International Edition. Pearson
  Education, Limited.

\bibitem[{Hampel et~al.(1986)Hampel, Ronchetti, Rousseeuw, and
  Stahel}]{hampel1986robust}
Hampel, F., Ronchetti, E., Rousseeuw, P., Stahel, W., 1986. Robust statistics.
  Wiley \& Sons, New York.

\bibitem[{Hastie et~al.(2001)Hastie, Tibshirani, and
  Friedman}]{hastie01statisticallearning}
Hastie, T., Tibshirani, R., Friedman, J., 2001. The Elements of Statistical
  Learning. Springer Series in Statistics. Springer New York Inc., New York,
  NY, USA.

\bibitem[{Hubert et~al.(2016)Hubert, Reynkens, Schmitt, and
  Verdonck}]{hubert2016sparse}
Hubert, M., Reynkens, T., Schmitt, E., Verdonck, T., 2016. Sparse {PCA} for
  high-dimensional data with outliers. Technometrics 58~(4), 424--434.

\bibitem[{Hubert et~al.(2005)Hubert, Rousseeuw, and
  Vanden~Branden}]{hubert2005robpca}
Hubert, M., Rousseeuw, P.~J., Vanden~Branden, K., 2005. {ROBPCA}: a new
  approach to robust principal component analysis. Technometrics 47~(1),
  64--79.

\bibitem[{Jolliffe(2002)}]{jolliffe2002}
Jolliffe, I., 2002. Principal component analysis. Wiley Online Library.

\bibitem[{Komsta and Novomestky(2015)}]{moments}
Komsta, L., Novomestky, F., 2015. {moments}: Moments, cumulants, skewness,
  kurtosis and related tests. R package version 0.14.
\newline\urlprefix\url{https://CRAN.R-project.org/package=moments}

\bibitem[{Nordhausen et~al.(2008)Nordhausen, Oja, and Tyler}]{ICS}
Nordhausen, K., Oja, H., Tyler, D.~E., 2008. Tools for exploring multivariate
  data: The package {ICS}. Journal of Statistical Software 28~(6), 1--31.

\bibitem[{Nordhausen et~al.(2016)Nordhausen, Oja, and Tyler}]{nordhausen2016}
Nordhausen, K., Oja, H., Tyler, D.~E., 2016. Asymptotic and bootstrap tests for
  subspace dimension. Tech. rep., arXiv:1611.04908v1.

\bibitem[{{Nordhausen} et~al.(2017){Nordhausen}, {Oja}, {Tyler}, and
  {Virta}}]{NordhausenOjaTylerVirta2017}
{Nordhausen}, K., {Oja}, H., {Tyler}, D.~E., {Virta}, J., 2017. {Asymptotic and
  bootstrap tests for the dimension of the non-Gaussian subspace}. Signal
  Processing Letters 24, 887--891.

\bibitem[{Nordhausen and Tyler(2015)}]{NordhausenTyler2015}
Nordhausen, K., Tyler, D.~E., 2015. A cautionary note on robust covariance
  plug-in methods. Biometrika 102, 573--588.

\bibitem[{Pe{\~n}a and Prieto(2001)}]{pena2001multivariate}
Pe{\~n}a, D., Prieto, F.~J., 2001. Multivariate outlier detection and robust
  covariance matrix estimation. Technometrics 43~(3).

\bibitem[{Penny and Jolliffe(1999)}]{penny1999multivariate}
Penny, K.~I., Jolliffe, I.~T., 1999. Multivariate outlier detection applied to
  multiply imputed laboratory data. Statistics in Medicine 18~(14), 1879--1895.

\bibitem[{Peres-Neto et~al.(2005)Peres-Neto, Jackson, and
  Somers}]{PeresNetoJacksonSomers2005}
Peres-Neto, P.~R., Jackson, D.~A., Somers, K.~M., 2005. How many principal
  components? {S}topping rules for determining the number of non-trivial axes
  revisited. Computational Statistics \& Data Analysis 49~(4), 974--997.

\bibitem[{{R Core Team}(2014)}]{R}
{R Core Team}, 2014. R: A Language and Environment for Statistical Computing. R
  Foundation for Statistical Computing, Vienna, Austria.
\newline\urlprefix\url{http://www.R-project.org/}

\bibitem[{Rousseeuw et~al.(2017)Rousseeuw, Croux, Todorov, Ruckstuhl,
  Salibian-Barrera, Verbeke, Koller, and M\"achler}]{robustbase}
Rousseeuw, P., Croux, C., Todorov, V., Ruckstuhl, A., Salibian-Barrera, M.,
  Verbeke, T., Koller, M., M\"achler, M., 2017. robustbase: Basic Robust
  Statistics. R package version 0.92-5.
\newline\urlprefix\url{http://CRAN.R-project.org/package=robustbase}

\bibitem[{Rousseeuw(1986)}]{Rousseeuw1986}
Rousseeuw, P.~J., 1986. Multivariate estimation with high breakdown point. In:
  Grossman, W., Pflug, G., Vincze, I., Wertz, W. (Eds.), Mathematical
  Statistics and Applications. Reidel, Dordrecht, pp. 283--297.

\bibitem[{Rousseeuw and Bossche(2017)}]{rousseeuw2016detecting}
Rousseeuw, P.~J., Bossche, W. V.~D., 2017. Detecting deviating data cells.
  Technometrics, 1--11.

\bibitem[{Rousseeuw and Van~Zomeren(1990)}]{rousseeuw1990unmasking}
Rousseeuw, P.~J., Van~Zomeren, B.~C., 1990. Unmasking multivariate outliers and
  leverage points. Journal of the American Statistical Association 85~(411),
  633--639.

\bibitem[{Stahel and M\"achler(2013)}]{robustX}
Stahel, W., M\"achler, M., 2013. robustX: eXperimental Functionality for Robust
  Statistics. R package version 1.1-4.
\newline\urlprefix\url{http://CRAN.R-project.org/package=robustX}

\bibitem[{Stahel and M\"achler(2009)}]{StahelMachler2009}
Stahel, W.~A., M\"achler, M., 2009. Comment on ``invariant co-ordinate
  selection''. Journal of the Royal Statistical Society B 71~(584--586).

\bibitem[{Tarr et~al.(2016)Tarr, M{\"u}ller, and Weber}]{tarr2016}
Tarr, G., M{\"u}ller, S., Weber, N.~C., 2016. Robust estimation of precision
  matrices under cellwise contamination. Computational Statistics \& Data
  Analysis 93, 404--420.

\bibitem[{Todorov and Filzmoser(2009)}]{rrcov}
Todorov, V., Filzmoser, P., 2009. An object-oriented framework for robust
  multivariate analysis. Journal of Statistical Software 32~(3), 1--47.

\bibitem[{Tyler et~al.(2009)Tyler, Critchley, D\"umbgen, and Oja}]{Tyler2009}
Tyler, D.~E., Critchley, F., D\"umbgen, L., Oja, H., 2009. Invariant coordinate
  selection. Journal of the Royal Statistical Society: Series B (Statistical
  Methodology) 71~(3), 549--592.

\bibitem[{Willems et~al.(2009)Willems, Joe, and Zamar}]{willems2009}
Willems, G., Joe, H., Zamar, R., 2009. Diagnosing multivariate outliers
  detected by robust estimators. Journal of Computational and Graphical
  Statistics 18~(1), 73--91.

\bibitem[{Yazici and Yolacan(2007)}]{Yazici2007}
Yazici, B., Yolacan, S., 2007. A comparison of various tests of normality.
  Journal of Statistical Computation and Simulation 77~(2), 175--183.

\end{thebibliography}




%%\bibitem[ ()]{}

%%\end{thebibliography}
%
%\centerline{{\sc Supplementary material}}
%The four supplemental files are contained in a single archive.
%\begin{description}
%\item[readme.txt:] brief description of the contents of the supplemental files.
%%\item[supp\_mat.pdf:] details of the proof of Proposition 1.
%\item[Scatterplot\_simulations.pdf:] figure with six scatterplot matrices to visualize the distribution of one data set from each of the Cases 0 to 5 and with $p=6$.
%\item[CodeR\_simulations\_functions.r:] R code to generate the simulated data (Cases 1 to 5).
%\item[CodeR\_examples.r:] R code to derive the results of Table 4 for the Reliability data and the HTP data sets.
%\end{description}

\appendix
\section{Proof of Proposition 1}\label{proofprop1}
Let us denote by $\bo M$ the $p\times q$ matrix whose columns contain the vectors $\bs \mu_h$, $h=1,\ldots,q$.
Given the affine invariance property of ICS, we assume w.l.o.g. that
$\bs\mu_0=\bo 0$, that $\bs\Sigma_W=\bo I_p$ where $\bo I_p$ denotes the $p\times p$ identity matrix and that the last $p-q$ rows of $\bo M$ contain zeros so that:
$\bo M=[\bo M_q,\bo 0]'$ where $\bo M_q$ is a $q\times q$ matrix. In the following, we also assume for convenience that the dimension of the vector space spanned by the columns
of $\bo M$ is $q$. Otherwise, we would have to reparametrize the mixture distribution with a number of clusters smaller than $q+1$ and equal to one plus the dimension of the
subspace spanned by the columns of $\bo M$. Under these assumptions, we determine that the total covariance matrix can be written as
$$\bs\Sigma=\left[\begin{array}{cc}
\bs\Sigma_q & \bo 0\\
\bo 0 & \bo I_{p-q}
\end{array}
\right] $$
where $\bs\Sigma_q$ denotes a non-singular $q\times q$ matrix.
We also denote by $\bo X_q$ (resp. $\bs \mu_{\bo X_q}$) the first $q$ rows of $\bo X$ (resp. of $\bs \mu_{\bo X}$).

Under the mixture distribution (\ref{model}), we have:
\begin{eqnarray*}
d^2(\bo X) & = & (\bo X_q - \bs \mu_{\bo X_q})'\bs \Sigma_q^{-1} (\bo X_q - \bs \mu_{\bo X_q}) + \sum_{i=q+1}^p X_{i}^2,\\
d^2_R(\bo X) & = & \sum_{i=1}^p X_{i}^2.
\end{eqnarray*}

We make use of the Lindeberg-Feller central limit theorem as recalled for instance in \cite{greene2012}, p.1119, which states that:

Let $Y_i$, $i=1,\ldots,n,$ be a sequence of independent random variables with finite means $m_i$ and finite positive variance $\sigma^2_i$.
Let
$$\bar{m}_n=\frac{1}{n} \sum_{i=1}^n m_i \;\; \mbox{ and } \;\; \bar{\sigma}^2_n=\frac{1}{n} \sum_{i=1}^n \sigma_i^2.$$
If $\displaystyle \lim_{n \rightarrow +\infty} \max(\sigma_i)/(n\bar{\sigma}_n) =0$ and
$\displaystyle \lim_{n \rightarrow +\infty} \bar{\sigma}_n^2 =\bar{\sigma}^2 < \infty$ then
$$ \sqrt{n} \left(\bar{Y}_n - \bar{m}_n \right) \underset{n \to +\infty}{\overset{d}{\longrightarrow}} {\cal N}(0,\bar{\sigma}^2)$$

We recall that $\bo X_{no}$ follows a normal distribution ${\cal N}(\bo 0,\bo I_p)$ and $\bo X_{o,h}$ follows a normal distribution
${\cal N}(\bs \mu_h,\bo I_p)$, with the last $p-q$ coordinates of $\bs \mu_h$ equal to 0 and  $h=1,\ldots,q$. We assume that $\bo X_{no}$ and $\bo X_{o,h}$,  for $h=1,\ldots,q$,
are independent, and we are interested in the behavior of the difference between the squared distance of $\bo X_{o}$ and of $\bo X_{no}$ for both Mahalanobis distances,
when dimension $p$ increases and $q$ is fixed.
We first look at the convergence in distribution of the difference of the robust Mahalanobis distances when $p$ grows to infinity and we have that
$$ d_R^2(\bo X_{o,h}) -d_R^2(\bo X_{no}) = \sum_{i=1}^p \left( X^2_{o,h,i} - X^2_{no,i}\right).$$
Let denote $Y_i=X^2_{o,h,i} - X^2_{no,i}$, $i=1,\ldots,p$, and check the Lindeberg-Feller theorem assumptions for this sequence when $p$ goes to infinity. Note that we apply the theorem to $p$ and not to $n$.
Given that the vectors $\bo X_{no}$ and $\bo X_{o,h}$ are Gaussian vectors with uncorrelated components and are independent between them, the random variables $Y_i$ are independent.
For $i=1,\ldots,p$, $X^2_{no,i}$ follows a chi-squared distribution with one degree of freedom and $X^2_{o,h,i}$ follows the same distribution for $i=q+1, \ldots,p$, while it follows a non central chi-squared distribution with one degree of freedom and noncentrality parameter $\mu_{h,i}^2$, for $i=1,\ldots,q$.
So the expectation of $Y_i$, $m_i=0$ for $i=q+1,\ldots,p$ and $m_i=\mu_{h,i}^2$ for $i=1,\ldots,q$.
The variance of $Y_i$ is finite, positive and equal $\sigma_i^2=4$ for $i=q+1,\ldots,p$ and $4 (\mu_{h,i}^2+1)$ for $i=1,\ldots,q$.
So $\bar{\sigma}_p^2=(4/p) [p+\sum_{i=1}^q \mu_{h,i}^2]$ and tends to 4 when $p$ goes to infinity.
Let denote $\sigma_{\max}^2= \max(\sigma^2_i)=4 (\max\{\mu_{h,i}^2, i=1,\ldots,q\}+1)$, which does not depend on $p$, then we have
$$ \lim_{p \rightarrow +\infty} \sigma_{\max}/(p\bar{\sigma}_p) = 0.$$
We conclude that:
$$ \frac{1}{\sqrt{p}} \left(d_R^2(\bo X_{o,h}) -d_R^2(\bo X_{no}) - \sum_{i=1}^q \mu^2_{h,i} \right) \underset{p \to +\infty}{\overset{d}{\longrightarrow}} {\cal N}(0,4).$$

\vspace{5mm}

\noindent For the non-robust Mahalanobis distance, we have:
\begin{eqnarray*}
d^2(\bo X_{o,h}) -d^2(\bo X_{no})  & = & (\bo X_{o,h,q} - \bs \mu_{\bo X_q})'\bs \Sigma_q^{-1} (\bo X_{o,h,q} - \bs \mu_{\bo X_q}) -
(\bo X_{no,q} - \bs \mu_{\bo X_q})'\bs \Sigma_q^{-1} (\bo X_{no,q} - \bs \mu_{\bo X_q})\\
 & &  + \sum_{i=q+1}^p \left( X^2_{o,h,i} - X^2_{no,i}\right).
 \end{eqnarray*}
where we denote by $\bo X_{o,h,q}$ (resp. $\bo X_{no,q}$) the first $q$ rows of $\bo X_{o,h}$ (resp. of $\bo X_{no}$).
$$\mbox{Let }\; Y_1= (\bo X_{o,h,q} - \bs \mu_{\bo X_q})'\bs \Sigma_q^{-1} (\bo X_{o,h,q} - \bs \mu_{\bo X_q}) -
(\bo X_{no,q} - \bs \mu_{\bo X_q})'\bs \Sigma_q^{-1} (\bo X_{no,q} - \bs \mu_{\bo X_q}) $$
and $Y_i=X^2_{o,h,q+i-1} - X^2_{no,q+i-1}$, for $i=2,\ldots,p-q+1$.
As previously, the $Y_i$s are independent.
As the difference of two non degenerate quadratic forms for $q$-dimensional Gaussian vectors, the expectation $m_1$ of $Y_1$ is finite and its variance $\sigma^2_1$  is finite and positive and does not depend on $p$.
For $i=2,\ldots,p-q+1$, $X^2_{no,i}$ and $X^2_{o,h,i}$ follow a chi-squared distribution with one degree of freedom and so the expectation
$m_i$ of $Y_i$ equals 0 and the variance $\sigma_i^2=4$.
So $\bar{m}_p=m_1/(p-q+1)$, $\bar{\sigma}_p^2=[4(p-q)+\sigma_1^2]/(p-q+1)$ and tends to 4 when $p$ goes to infinity.
Let $\sigma_{\max}^2= \max(\sigma^2_i)=\max(4,\sigma_1^2)$, which does not depend on $p$, then we have
$$ \lim_{p \rightarrow +\infty} \sigma_{\max}/(p\bar{\sigma}_p) = 0.$$
We conclude that:
$$ \frac{\sqrt{p}}{p-q+1} \left(d^2(\bo X_{o,h}) -d^2(\bo X_{no}) - m_1 \right) \underset{p \to +\infty}{\overset{d}{\longrightarrow}} {\cal N}(0,4)$$
which gives the final result.

\section{Derivation of the eigenvalues and eigenvectors of the simultaneous diagonalization of $\cov$ and $\cov_4$ for particular mixtures}\label{sec:CompChap2}
%% paragraphe à modifier
%As stated in Chapter~\ref{chap:ICS}, analyzing the Invariant Components resulting of the diagonalization of $\bo V_1= \cov$ and $\bo V_2 = \cov_4$ presents the advantage to be analytically tractable. Theorem 4 in \cgcol{black}\cite{Tyler2009} establishes that the structure of the data can be recovered in the first and/or the last components, even for different combinations of scatter matrices, under a mixture of $k$ elliptically symmetric distributions with possibly different location and scatter parameters. However, for outlier detection purpose it is of great importance to know if only the first components should be analyzed, or the last ones or both.\\
%Considering only \acrshort{ICS} with  $\bo V_1= \cov$ and $\bo V_2 = \cov_4$ enables us to derive the eigenvalues $\rho_1,\dots,\rho_p$ associated to the eigenvectors on which the data is projected. Based on their order, we can conclude on which components is of interest depending on the parameters of the model.
%On this note, we give explicit conditions for some of the cases introduced in Chapter~\ref{chap:ICS}, corroborating the fact we focus only on the first components.

Let $\bo X = (X_1, \dots ,X_p)'$ be a $p$-multivariate real random vector and denote by $\bo F_{\bo X}$ the distribution of $\bo X$. We assume that $p>2$ and that
$\bo F_{\bo X}$ admits fourth moments.
The functional versions of $\cov(\bo F_{\bo X}) $ and $\cov_4(\bo F_{\bo X}) $ which are consistent at the Gaussian distribution are given by:
$$ \cov(\bo F_{\bo X})=\mathbb{E}\left[(\bo X-\mathbb{E}(\bo X))(\bo X-\mathbb{E}(\bo X))'\right ],$$
$$\cov_4(\bo F_{\bo X})=\frac{1}{p+2}\,\mathbb{E}\left[(\bo X-\mathbb{E}(\bo X))'\cov^{-1}(\bo F_{\bo X})(\bo X-\mathbb{E}(\bo X))(\bo X-\mathbb{E}(\bo X))(\bo X-\mathbb{E}(\bo X))'\right ]. $$
We denote by $\rho_1(\bo F_{\bo X}) \geq \rho_2(\bo F_{\bo X}) \ldots \geq \rho_p(\bo F_{\bo X})$ the eigenvalues of $\cov^{-1}(\bo F_{\bo X})\cov_4(\bo F_{\bo X}) $ in decreasing
order.
The cases we consider below correspond or are very similar to Cases 1, 2 and 5 from the simulations section. For such mixtures and the scatter pair $\cov$-$\cov4$, it is possible to derive conditions under which the ICS method recovers the direction of outlying observations.

\subsection{Case 1:  mean-shift outlier model}
Let $\bo F_{\bo X}$ be a mixture of two Gaussian distributions with different location parameters and the same definite positive covariance matrix $\bs \Sigma_1$:
\begin{equation}\label{case1}
\bo X \sim (1-\epsilon)~ \mathcal{N}(\bo 0_p, \bs \Sigma_1) + \epsilon ~\mathcal{N}(\bs \mu, \bs \Sigma_1)
\end{equation}
with $\epsilon<0.5$ and $\bs \mu \neq \bo 0_p$ a $p$-vector.\\
In this case, the behavior of ICS has already been established. This result is explicitly presented in \cite{Tyler2009} as a particular case of the Theorem 3. \cite{caussinus1994} and \cite{caussinus1995metrics} also derived this condition as a particular case of the symmetrized version of the one-step $W$-estimate used as one of the scatter matrix while the other was the usual covariance matrix. Finally, \cite{Alashwali2016} also recovered the same result by using arguments from \cite{pena2001multivariate} focusing on projection pursuit based on the kurtosis. As a reminder, the result is the following.
\vspace{5pt}
\begin{prop}\label{prop:case1}~\\
Let $\bo X$ follow the distribution (\ref{case1}), the eigenvalues of $\cov^{-1}(\bo F_{\bo X})\cov_4(\bo F_{\bo X}) $ are such that either:\\
\begin{tabular}{ccccrl}
		(a) &	$\rho_1(\bo F_{\bo X})~ >$&$ \rho_2(\bo F_{\bo X})  $&$=$&$\dots ~~~~= ~\rho_p(\bo F_{\bo X})$  &if $\displaystyle \epsilon < (3 - \sqrt{3})/6$ ($\approx$ 21\%),\\
		(b) &	$\rho_1(\bo F_{\bo X})~=$&$\dots$&$=$&$ \rho_{p-1}(\bo F_{\bo X}) > ~\rho_p(\bo F_{\bo X})$ &if $\displaystyle \epsilon > (3 - \sqrt{3})/6$,\\
		(c) & $\rho_1(\bo F_{\bo X}) ~=$&$ \rho_2(\bo F_{\bo X}) $&$=$&$\dots ~~~~ =~ \rho_p(\bo F_{\bo X})$ &if  $\displaystyle \epsilon = (3 - \sqrt{3})/6$.
		\end{tabular}
\vspace{2mm}

Moreover, if (a) (resp. (b)) holds then the eigenvector of $\cov^{-1}(\bo F_{\bo X})\cov_4(\bo F_{\bo X}) $ associated with $\rho_1(\bo F_{\bo X})$ (resp. $\rho_p(\bo F_{\bo X})$)
is proportional to $\bs \Sigma_1^{-1} \bs \mu$.
\end{prop}
~
\begin{rem}
In the simulation framework, Case 1 corresponds to model (\ref{case1}) with a percentage of contamination equal to 2\% and outliers  are highlighted  on the first component of ICS.
\end{rem}

\subsection{Case 2a: the barrow wheel distribution}
This distribution was suggested by Stahel and M\"{a}chler in the discussion in \cite{Tyler2009} as a ``benchmark distribution for multivariate tools''. This so-called barrow wheel distribution was first introduced in \cite{hampel1986robust}. Here, we simplify the model without loss of generality by considering no rescaling nor rotation  because of the affine invariance property of the ICS method. \\
Let the distribution of $\bo X$ be:
\begin{equation}\label{case2}
\bo X \sim (1-\epsilon) ~\mathcal{N}(\bo 0_p, \bs \Sigma_{21} ) + \epsilon~ H
\end{equation}
where $\bs \Sigma_{21} = \diag(\sigma_{11}^2, 1, \dots,  1)$ and let $\bo Y=(Y_1,\ldots,Y_p)'$ distributed according to $H$. $H$ is such that $ Y_1$ has a symmetric distribution with $ Y_1^2 \sim \chi_k^2$ and is independent of $ Y_2, \dots,  Y_p \sim \mathcal{N}(\bo 0_p, \bs \Sigma_{22})$ with $\bs \Sigma_{22} = \sigma_{22}^2 \bo I_{p-1}$. With such a model, the outliers are generated along the first direction on both sides of the main data.\\

\cite{Tyler2009} prove in the discussion that their Theorem~4 is still valid under the barrow wheel distribution. Restricting the analysis to $\cov $ and $\cov_4$ enables us to derive a more precise result.
\vspace{5pt}
\begin{prop}\label{prop2}~\\
Let $\bo X$ follow the distribution (\ref{case2}), the eigenvalues of $\cov^{-1}(\bo F_{\bo X})\cov_4(\bo F_{\bo X}) $ are such that either:\\
	\begin{tabular}{ccccr}
		(a) &	$\rho_1(\bo F_{\bo X})~ >$&$ \rho_2(\bo F_{\bo X})  $&$=$&$\dots ~~~~= ~\rho_p(\bo F_{\bo X})$ ,\\
		(b) &	$\rho_1(\bo F_{\bo X})~=$&$\dots$&$=$&$ \rho_{p-1}(\bo F_{\bo X}) > ~\rho_p(\bo F_{\bo X})$,\\
		(c) & $\rho_1(\bo F_{\bo X}) ~=$&$ \rho_2(\bo F_{\bo X}) $&$=$&$\dots ~~~~ =~ \rho_p(\bo F_{\bo X})$.
		\end{tabular}
		~\\
\vspace{5pt}
		~\\
with $\rho_1(\bo F_{\bo X}) = \dfrac{1}{p+2}\left(\dfrac{3(1-\epsilon)\sigma_{11}^4+\epsilon(2+k)k}{((1-\epsilon)\sigma_{11}^2+\epsilon k)^2}+p-1\right)$\\
~\\
\vspace{5pt}
and $\rho_2(\bo F_{\bo X}) =\dfrac{1}{p+2}\left(\dfrac{3((1-\epsilon)+\epsilon \sigma_{22}^4)}{((1-\epsilon)+\epsilon \sigma_{22}^2)^2}+p-1\right) $.

Moreover, if (a) (resp. (b))  holds then the eigenvector of $\cov^{-1}(\bo F_{\bo X})\cov_4(\bo F_{\bo X}) $ associated with $\rho_1(\bo F_{\bo X})$
(resp. with $\rho_p(\bo F_{\bo X})$) is proportional to $\bo e_1=(1,0,\ldots,0)'$.
\end{prop}

\begin{rem}
In the simulation framework, Case 2 corresponds to model (\ref{case2}) with $k=5$, $\sigma_{11}^2 = 0.1$, $\sigma_{11}^2 = 0.2$ and $\epsilon=2\%$.
 In this situation, $\rho_1(\bo F_{\bo X})~ > \rho_2(\bo F_{\bo X})$ and the outliers are highlighted  on the first component of ICS.
\end{rem}

\begin{proof}\label{proof2}~Let us  compute the eigenvalues of $\cov(\bo F_{\bo X})^{-1} \cov_4(\bo F_{\bo X})$.

\noindent{\textsl{Moments of $Y_1$:}\\}
We can decompose $Y_1$ as $Y_1 = S\, C$, with
  $S = \left\{ \begin{tabular}{cc}
        -1 & with probability 1/2 \\
        1  & with probability 1/2
     \end{tabular} \right.$
		and $C\sim \chi_k$.\\
So, $\mathbb{E}(Y_1)=0$, $\var(Y_1)=\mathbb{E}(Y_1^2)=\mathbb{E}(C^2)=k$ and  $\var(Y_1)=\mathbb{E}(Y_1^2)=2k$.\\

\noindent{\textsl{Computation of $\cov(\bo F_{\bo X})^{-1}$:}\\}
For model (\ref{case2}), the expectation is $\mathbb{E}(\bo X) = \bo 0_p$, the between covariance is $\bs \Sigma_B = \bo 0$ and the within covariance matrix is $\bs \Sigma_W = \diag((1-\epsilon)\sigma_{11}^2+\epsilon k, ((1-\epsilon)+\epsilon \sigma_{22}^2), \ldots)$. So,
$$\cov(\bo F_{\bo X}) = \begin{pmatrix} \gamma_1 & \bo 0\\ \bo 0 & \gamma_2\bo I_{p-1} \end{pmatrix} \text{~~and~~}
\cov(\bo F_{\bo X})^{-1} = \begin{pmatrix} 1/\gamma_1 & \bo 0\\ \bo 0 & 1/\gamma_2\bo I_{p-1} \end{pmatrix},$$
with $\gamma_1 = (1-\epsilon)\sigma_{11}^2+\epsilon k$ and $\gamma_2 = (1-\epsilon)+\epsilon \sigma_{22}^2$.\\

\noindent{\textsl{Computation of $\cov_4(\bo F_{\bo X})$:}\\}
The scatter matrix  based on the fourth moments $\cov_4$ is defined by: $$\cov_4(\bo F_{\bo X}) =  \frac{1}{(p+2)} \mathbb{E}( d^2 ( \bo X - \mathbb{E}(\bo X) )(\bo X -\mathbb{E}(\bo X) )')$$
where $d^2 = d(\bo X)^2 = ||\cov(\bo F_{\bo X})^{-1/2} (\bo X -\mathbb{E}(\bo X))||^2$ is the classical squared Mahalanobis distance.
Here, $\mathbb{E}(\bo X) = \bo 0_p$ so,  $$\cov_4(\bo F_{\bo X}) =\frac{1}{(p+2)} \, \diag(\mathbb{E}(d^2 X_1^2), \dots, \mathbb{E}(d^2 X_p^2))$$ as all the $X_i$ are independent and $d^2=  \tfrac{1}{\gamma_1} X_{1}^2+  \tfrac{1}{\gamma_2}\sum_{l=2}^{p} X_{l}^2$.

\noindent{The first diagonal term is $\mathbb{E}(d^2 X_1^2) = \tfrac{1}{\gamma_1} \mathbb{E}(X_1^4) + (p-1)\gamma_1$.\\}
$\mathbb{E}(X_1^4) $ can be easily expressed since $X_1 \sim (1-\epsilon) Z_1+ \epsilon Y_1$ with $Z_1 \sim  \mathcal{N}(0, \sigma_{11}^2 )$ and $\mathbb{E}( Y_1^4)= \var(Y_1^2)+\mathbb{E}( Y_1^2)^2=(2+k)k$. Then, we  apply the following properties to have an expression for $\mathbb{E}(X_1^4)$:
\begin{itemize}
	\item Additive property of the moments:\\ $\mathbb{E}(X_1^4)  = (1-\epsilon) \mathbb{E}(Z_1^4)+  \epsilon \mathbb{E}(Y_1^4)$
	\item Decomposition of a fourth order moment:\\ $\mathbb{E}(Z_i^4)= \mathbb{E}((Z_i - \mathbb{E}(Z_i))^4) + 4 \mathbb{E}((Z_i - \mathbb{E}(Z_i))^3) \mathbb{E}(Z_i) +6 \mathbb{E}((Z_i - \mathbb{E}(Z_i))^2)\mathbb{E}(Z_i)^2 +4 \mathbb{E}(Z_i - \mathbb{E}(Z_i)) \mathbb{E}(Z_i)^3+ \mathbb{E}(Z_i)^4$.
	\item Computation of moments from a Gaussian distribution: \\
	If $Z \sim \mathcal{N}(\mu, \sigma^2)$, then $ \mathbb{E}((Z - \mu)^{2k})=\left((2k)!\sigma^{2k}\right)/(2^kk!)$ and $ \mathbb{E}((Z - \mu)^{2k+1})=0$.
\end{itemize}
So, $\mathbb{E}(X_1^4)= 3(1-\epsilon)\sigma_{11}^4+\epsilon (2+k)k$.\\
Finally, $\mathbb{E}(d^2 X_1^2) = \tfrac{1}{\gamma_1} ( 3(1-\epsilon)\sigma_{11}^4+\epsilon (2+k)k) + (p-1)\gamma_1$.

\noindent{All the other diagonal terms are equal to $\mathbb{E}(d^2 X_j^2) = \tfrac{1}{\gamma_2} \mathbb{E}(X_j^4) + (p-1)\gamma_2$ for $j=2,\dots,p$.\\}
Since $X_j \sim (1-\epsilon) Z_1 + \epsilon Z_2$ with $Z_1 \sim  \mathcal{N}(0, 1 )$ and $Z_2\sim \mathcal{N}(0,  \sigma_{22}^2)$, we can apply the same procedure as previously and we obtain: $\mathbb{E}(X_j^4) = 3((1-\epsilon)+\epsilon \sigma_{22}^4)$ and so, for $j=2,\dots,p$, $\mathbb{E}(d^2 X_j^2) = \tfrac{1}{\gamma_2}(3((1-\epsilon)+\epsilon \sigma_{22}^4)) + (p-1)\gamma_2$.\\

\noindent{\textsl{Computation of $\cov(\bo F_{\bo X})^{-1}\cov_4(\bo F_{\bo X})$:}\\}
Now we can express $\cov(\bo F_{\bo X})^{-1}\cov_4(\bo F_{\bo X})$ as:
$$\cov(\bo F_{\bo X})^{-1}\cov_4(\bo F_{\bo X}) = \frac{1}{(p+2)} \begin{pmatrix}
\tfrac{1}{\gamma_1}\mathbb{E}(d^2 X_1^2)& \bo 0\\
\bo 0 &
 \tfrac{1}{\gamma_2}\mathbb{E}(d^2 X_j^2) \bo I_{p-1}\end{pmatrix} $$
So, the eigenvalues of $\cov(\bo F_{\bo X})^{-1}\cov_4(\bo F_{\bo X})$ are simply its diagonal terms and the eigenvector associated with $\mathbb{E}(d^2 X_1^2)/ \gamma_1$
is $\bo e_1$.\\
If $\tfrac{1}{\gamma_1}\mathbb{E}(d^2 X_1^2) > \tfrac{1}{\gamma_2}\mathbb{E}(d^2 X_j^2)$ then $\rho_1(\bo F_{\bo X}) >  \rho_2(\bo F_{\bo X})$,\\
$$\mbox{with } \rho_1(\bo F_{\bo X}) = \dfrac{1}{p+2}\left(\dfrac{3(1-\epsilon)\sigma_{11}^4+\epsilon(2+k)k}{((1-\epsilon)\sigma_{11}^2+\epsilon k)^2}+p-1\right),$$

$$\rho_2(\bo F_{\bo X}) =\dfrac{1}{p+2}\left(\dfrac{3((1-\epsilon)+\epsilon \sigma_{22}^4)}{((1-\epsilon)+\epsilon \sigma_{22}^2)^2}+p-1\right)$$

and the eigenvector associated with $\mathbb{E}(d^2 X_1^2)/ \gamma_1$ is $\bo e_1$.
And so Proposition \ref{prop2} is proven.
\end{proof}

The eigenvalues expression can be easily simplified in the case when $\bs \Sigma_{21}=\bs \Sigma_{22}$ and so we derive the following corollary.

\begin{coro}
If $\bo X$ follows the distribution (\ref{case2}) with $\sigma_{11}=\sigma_{22}=1$,  the eigenvalues of $\cov^{-1}(\bo F_{\bo X})\cov_4(\bo F_{\bo X}) $ are such that either:

	\begin{tabular}{ccccrl}
(a) &	$\rho_1(\bo F_{\bo X})~ >$&$ \rho_2(\bo F_{\bo X})  $&$=$&$\dots ~~~~= ~\rho_p(\bo F_{\bo X})$  &if $\epsilon < (k-3)/(3(k-1))$,\\
		(b) &	$\rho_1(\bo F_{\bo X})~=$&$\dots$&$=$&$ \rho_{p-1}(\bo F_{\bo X}) > ~\rho_p(\bo F_{\bo X})$ &if $\epsilon > (k-3)/(3(k-1))$,\\
		(c) & $\rho_1(\bo F_{\bo X}) ~=$&$ \rho_2(\bo F_{\bo X}) $&$=$&$\dots ~~~~ =~ \rho_p(\bo F_{\bo X})$ &if  $\epsilon = (k-3)/(3(k-1))$.
\end{tabular}

\vspace{1mm}

Moreover, if (a) (resp. (b))  holds then the eigenvector of $\cov^{-1}(\bo F_{\bo X})\cov_4(\bo F_{\bo X}) $ associated with $\rho_1(\bo F_{\bo X})$
(resp. $\rho_p(\bo F_{\bo X})$) is proportional to $\bo e_1$.
\end{coro}
The bound $(k-3)/(3(k-1))$ on the contamination is minimum for $k=4$ and equals $1/9$. It increases with $k$ and its limit equals to $1/3\simeq 33\%$ of contamination when $k$ grows to infinity.

\subsection{Case 2b: symmetric contamination of a Gaussian distribution}
We can also mimic the barrow wheel distribution by the following mixture of three Gaussian distributions:
\begin{equation}\label{case2bis}
\bo X \sim (1-\epsilon)~ \mathcal{N}(\bo 0_p, \bs \Sigma_{21} ) + \tfrac{\epsilon}{2}~ \mathcal{N}(\delta \bo e_1, \bs \Sigma_{22}) +\tfrac{\epsilon}{2}~ \mathcal{N}(-\delta \bo e_1, \bs \Sigma_{22})
\end{equation}
with $\bs \Sigma_{21} = \diag(\sigma_{11}^2, \sigma_{12}^2, \dots,   \sigma_{12}^2)$, $\bs \Sigma_{22} = \diag(\sigma_{21}^2, \sigma_{22}^2, \dots,   \sigma_{22}^2)$ and $\delta \neq 0$. \\
In this case, we can derive the following proposition.
\vspace{5pt}
\begin{prop}\label{prop2bis}~\\
Let $\bo X$ follow the distribution (\ref{case2bis}), the eigenvalues of $\cov^{-1}(\bo F_{\bo X})\cov_4(\bo F_{\bo X}) $ are such that either:\\
	\begin{tabular}{ccccr}
		(a) &	$\rho_1(\bo F_{\bo X})~ >$&$ \rho_2(\bo F_{\bo X})  $&$=$&$\dots ~~~~= ~\rho_p(\bo F_{\bo X})$ ,\\
		(b) &	$\rho_1(\bo F_{\bo X})~=$&$\dots$&$=$&$ \rho_{p-1}(\bo F_{\bo X}) > ~\rho_p(\bo F_{\bo X})$,\\
		(c) & $\rho_1(\bo F_{\bo X}) ~=$&$ \rho_2(\bo F_{\bo X}) $&$=$&$\dots ~~~~ =~ \rho_p(\bo F_{\bo X})$.
		\end{tabular}
		~\\
\vspace{5pt}
		~\\
with $\rho_1(\bo F_{\bo X}) = \dfrac{1}{p+2}\left(\dfrac{3(1-\epsilon)\sigma_{11}^4+\epsilon(3\sigma_{21}^4+6\sigma_{21}^2 \delta^2+\delta^4)}{((1-\epsilon)\sigma_{11}^2+\epsilon(\sigma_{21}^2+\delta^2))^2}+p-1\right)$\\
~\\
\vspace{5pt}
and $\rho_2(\bo F_{\bo X}) =\dfrac{1}{p+2}\left(\dfrac{3((1-\epsilon)\sigma_{12}^4+\epsilon \sigma_{22}^4)}{((1-\epsilon)\sigma_{12}^2+\epsilon \sigma_{22}^2)^2}+p-1\right) $.

Moreover, if (a) (resp. (b))  holds then the eigenvector of $\cov^{-1}(\bo F_{\bo X})\cov_4(\bo F_{\bo X}) $ associated with $\rho_1(\bo F_{\bo X})$
(resp. $\rho_p(\bo F_{\bo X})$) is proportional to $\bo e_1$.
\end{prop}

\begin{proof}\label{proof2bis}~Let us  compute the eigenvalues of $\cov(\bo F_{\bo X})^{-1} \cov_4(\bo F_{\bo X})$.

\noindent{\textsl{Computation of $\cov(\bo F_{\bo X})^{-1}$:}\\}
For the model (\ref{case2bis}), the expectation is $\mathbb{E}(\bo X) = \bo 0_p$, the within covariance matrix is $\bs \Sigma_W = (1-\epsilon) \bs \Sigma_{21} + \epsilon \bs \Sigma_{22} = \diag((1-\epsilon)\sigma_{11}^2+\epsilon \sigma_{21}^2, ((1-\epsilon)\sigma_{12}^2+\epsilon \sigma_{22}^2), \dots)$ and the between covariance is $\bs \Sigma_B = \epsilon \delta^2 \bo e_1 \bo e_1'$. So,
$$\cov(\bo F_{\bo X}) = \begin{pmatrix} \gamma_1 & \bo 0\\ \bo 0 & \gamma_2\bo I_{p-1} \end{pmatrix} \text{~~and~~}
\cov(\bo F_{\bo X})^{-1} = \begin{pmatrix} 1/\gamma_1 & \bo 0\\ \bo 0 & 1/\gamma_2\bo I_{p-1} \end{pmatrix},$$
with $\gamma_1 = (1-\epsilon)\sigma_{11}^2+\epsilon \sigma_{21}^2 +  \epsilon \delta^2$ and $\gamma_2 = (1-\epsilon)\sigma_{12}^2+\epsilon \sigma_{22}^2$.

\noindent{\textsl{Computation of $\cov_4(\bo F_{\bo X})$:}\\}
As already defined in Proof~\ref{proof2},
  $$\cov_4(\bo F_{\bo X}) =\frac{1}{(p+2)} \diag(\mathbb{E}(d^2 X_1^2), \dots, \mathbb{E}(d^2 X_p^2))$$ where  $d^2=  \tfrac{1}{\gamma_1} X_{1}^2+  \tfrac{1}{\gamma_2}\sum_{l=2}^{p} X_{l}^2$.

\noindent{The first diagonal term is $\mathbb{E}(d^2 X_1^2) = \tfrac{1}{\gamma_1} \mathbb{E}(X_1^4) + (p-1)\gamma_1$.\\}
$\mathbb{E}(X_1^4) $ can be easily expressed since $X_1 \sim (1-\epsilon) Z_1+ \tfrac{\epsilon}{2} Z_2 +\tfrac{\epsilon}{2} Z_3$ with $Z_1 \sim  \mathcal{N}(0, \sigma_{11}^2 )$, $Z_2\sim \mathcal{N}(\delta,  \sigma_{21}^2)$ and $Z_3\sim \mathcal{N}(-\delta, \sigma_{21}^2)$. Then, we  apply the same properties as in Proof~\ref{proof2} and so we obtain $\mathbb{E}(X_1^4)= 3(1-\epsilon)\sigma_{11}^4+\epsilon(3\sigma_{21}^4+6\sigma_{21}^2 \delta^2+\delta^4)$.\\
Finally, $\mathbb{E}(d^2 X_1^2) = \tfrac{1}{\gamma_1} ( 3(1-\epsilon)\sigma_{11}^4+\epsilon(3\sigma_{21}^4+6\sigma_{21}^2 \delta^2+\delta^4)) + (p-1)\gamma_1$.

\noindent{All the other diagonal terms are equal to $\mathbb{E}(d^2 X_j^2) = \tfrac{1}{\gamma_2} \mathbb{E}(X_j^4) + (p-1)\gamma_2$ for $j=2,\dots,p$.\\}
Since $X_j \sim (1-\epsilon) Z_1 + \epsilon Z_2$ with $Z_1 \sim  \mathcal{N}(0, \sigma_{12}^2 )$ and $Z_2\sim \mathcal{N}(0,  \sigma_{22}^2)$, we can apply the same procedure as previously and we obtain: $\mathbb{E}(X_j^4) = 3((1-\epsilon)\sigma_{12}^4+\epsilon \sigma_{22}^4)$ and so, for $j=2,\dots,p$, $\mathbb{E}(d^2 X_j^2) = \tfrac{1}{\gamma_2}(3((1-\epsilon)\sigma_{12}^4+\epsilon \sigma_{22}^4)) + (p-1)\gamma_2$.

\noindent{\textsl{Computation of $\cov(\bo F_{\bo X})^{-1}\cov_4(\bo F_{\bo X})$:}\\}
Now we can express $\cov(\bo F_{\bo X})^{-1}\cov_4(\bo F_{\bo X})$ as:
$$\cov(\bo F_{\bo X})^{-1}\cov_4(\bo F_{\bo X}) = \frac{1}{(p+2)} \begin{pmatrix}
\tfrac{1}{\gamma_1}\mathbb{E}(d^2 X_1^2)& \bo 0\\
\bo 0 &
 \tfrac{1}{\gamma_2}\mathbb{E}(d^2 X_j^2) \bo I_{p-1}\end{pmatrix} $$
So, the eigenvalues of $\cov(\bo F_{\bo X})^{-1}\cov_4(\bo F_{\bo X})$ are simply its diagonal terms and the eigenvector associated with $\mathbb{E}(d^2 X_1^2)/ \gamma_1$
is $\bo e_1$.\\
If $\tfrac{1}{\gamma_1}\mathbb{E}(d^2 X_1^2) > \tfrac{1}{\gamma_2}\mathbb{E}(d^2 X_j^2)$ then $\rho_1(\bo F_{\bo X}) >  \rho_2(\bo F_{\bo X})$,\\
\centerline{
  with: $\rho_1(\bo F_{\bo X}) = \dfrac{1}{p+2}\left(\dfrac{3(1-\epsilon)\sigma_{11}^4+\epsilon(3\sigma_{21}^4+6\sigma_{21}^2 \delta^2+\delta^4)}{((1-\epsilon)\sigma_{11}^2+\epsilon(\sigma_{21}^2+\delta^2))^2}+p-1\right)$ }\\
~\\
\vspace{5pt}
\centerline{and $\rho_2(\bo F_{\bo X}) =\dfrac{1}{p+2}\left(\dfrac{3((1-\epsilon)\sigma_{12}^4+\epsilon \sigma_{22}^4)}{((1-\epsilon)\sigma_{12}^2+\epsilon \sigma_{22}^2)^2}+p-1\right)$. }\\
And so Proposition \ref{prop2bis} is proven.
\end{proof}

The eigenvalues expression above can be easily simplified in the case of equal covariance matrices $\bs \Sigma_{21}=\bs \Sigma_{22}$. Moreover, in this case, given that the ICS method is affine invariant, we do not need to assume diagonal matrices. We have thus the following corollary where the condition for the three cases depends only on the percentage of contamination $\epsilon$ and not on the location parameter $\delta$.

\begin{coro}
 For $\bo X$ simulated as in (\ref{case2bis}) with $\bs \Sigma_{21} = \bs \Sigma_{22}$ but not necessarily diagonal, the eigenvalues of $\cov(\bo F_{\bo X})^{-1} \cov_4(\bo F_{\bo X})$ are such that either:\\
	\begin{tabular}{ccccrl}
		(a) &	$\rho_1(\bo F_{\bo X})~ >$&$ \rho_2(\bo F_{\bo X})  $&$=$&$\dots ~~~~= ~\rho_p(\bo F_{\bo X})$  &if $\epsilon < 1/3$,\\
		(b) &	$\rho_1(\bo F_{\bo X})~=$&$\dots$&$=$&$ \rho_{p-1}(\bo F_{\bo X}) > ~\rho_p(\bo F_{\bo X})$ &if $\epsilon > 1/3$,\\
		(c) & $\rho_1(\bo F_{\bo X}) ~=$&$ \rho_2(\bo F_{\bo X}) $&$=$&$\dots ~~~~ =~ \rho_p(\bo F_{\bo X})$ &if  $\epsilon = 1/3$.
		\end{tabular}

Moreover, if (a) (resp. (b))  holds then the eigenvector of $\cov^{-1}(\bo F_{\bo X})\cov_4(\bo F_{\bo X}) $ associated with $\rho_1(\bo F_{\bo X})$
(resp. $\rho_p(\bo F_{\bo X})$) is proportional to $\bo e_1$.
\end{coro}

\subsection{Case 5: scale-shift outlier model}
Let $\bo X = (\bo X_1, \dots ,\bo X_p)'$ be a $p$-multivariate real random vector and assume the distribution
of $\bo X$ is a mixture of two Gaussian distributions with the same location parameters but with a scale change:
\begin{equation}\label{case5}
\bo X \sim (1-\epsilon) \mathcal{N}(\bo 0_p, \bo I_p) + \epsilon \mathcal{N}(\bo 0_p, \bs \Sigma_5)
\end{equation}
with $\epsilon<0.5$, $\bs \Sigma_5= \diag(\alpha \bo I_q, \bo I_{p-q})$, $q<p$ and $\alpha >1$.\\
This model generates outliers in up to $q$ directions via a scale-shift. In this context, the following proposition arises.

\vspace{5pt}
\begin{prop}\label{prop:case5}~\\
Let $\bo X$ follow the distribution (\ref{case5}), then the eigenvalues of $\cov^{-1}(\bo F_{\bo X})\cov_4(\bo F_{\bo X}) $ are such that:
$$\rho_1(\bo F_{\bo X}) = \dots = \rho_q(\bo F_{\bo X}) > \rho_{q+1}(\bo F_{\bo X})  =\dots ~~~~= ~\rho_p(\bo F_{\bo X})$$
%\begin{tabular}{ccccrl}
%		(a) &	$\rho_1(\bo F_{\bo X})~ = \dots = \rho_q(\bo F_{\bo X}) >$&$ \rho_{q+1}(\bo F_{\bo X})  $&$=$&$\dots ~~~~= ~\rho_p(\bo F_{\bo X})$  &if $\alpha > 1$,\\
%		(b) & $\rho_1(\bo F_{\bo X}) ~= \dots = \rho_q(\bo F_{\bo X}) =$&$ \rho_{q+1}(\bo F_{\bo X}) $&$=$&$\dots ~~~~ =~ \rho_p(\bo F_{\bo X})$ &if  $\alpha = 1$.
%		\end{tabular}

Moreover, the eigenvectors of $\cov^{-1}(\bo F_{\bo X})\cov_4(\bo F_{\bo X}) $ associated with the $q$ largest eigenvalues span the subspace spanned by $\{\bo e_1,\ldots,\bo e_q\}$.
\end{prop}

\noindent{Note that if $q=p$ then all the eigenvalues are equal and ICS is not informative and leads to the Mahalanobis distance.}
\vspace{5pt}
\begin{rem}
In the simulation framework, Case 5 is similar to model (\ref{case5}) with $\alpha=5$ and thus if $p>6$, $\rho_1(\bo F_{\bo X}) = \dots = \rho_6(\bo F_{\bo X}) > \rho_{7}(\bo F_{\bo X})  =\dots ~~~~= ~\rho_p(\bo F_{\bo X})$ and the outliers are highlighted  on the first six components, independently of the percentage of contamination~$\epsilon$.
\end{rem}

\begin{proof}~Let us  compute the eigenvalues of $\cov(\bo F_{\bo X})^{-1} \cov_4(\bo F_{\bo X})$.

\noindent{\textsl{Computation of $\cov(\bo F_{\bo X})^{-1}$:}\\}
For the model (\ref{case5}), the expectation is $\mathbb{E}(\bo X) = \bo 0_p$, the within covariance matrix is
 $\bs \Sigma_W = \diag(\gamma_1 \bo I_q,  \bo I_{p-q})$, with $\gamma_1=(1-\epsilon)+\alpha \epsilon$,
and the between covariance is $\bs \Sigma_B = \bo 0_p$. So,
$$\cov(\bo F_{\bo X}) = \bs \Sigma_W  \begin{pmatrix} \gamma_1 \bo I_q & \bo 0\\ \bo 0 &\bo I_{p-q} \end{pmatrix} \text{~~and~~}
\cov(\bo F_{\bo X})^{-1} = \begin{pmatrix} 1/\gamma_1 \bo I_q & \bo 0\\ \bo 0 & \bo I_{p-q} \end{pmatrix},$$
with $\gamma_1 = (1-\epsilon)+\alpha \epsilon$.

\noindent{\textsl{Computation of $\cov_4(\bo F_{\bo X})$:}\\}
As already defined in Proof~\ref{proof2},
  $$\cov_4(\bo F_{\bo X}) =\frac{1}{(p+2)} \diag(\mathbb{E}(d^2 X_1^2), \dots, \mathbb{E}(d^2 X_p^2))$$ where  $d^2=  \tfrac{1}{\gamma_1} \sum_{l=1}^{q}  X_{l}^2+  \sum_{j=q+1}^{p} X_{j}^2$.

The first $q$ diagonal terms are equal to $\mathbb{E}(d^2 X_l^2) = \mathbb{E}(d^2 X_1^2)$ for $l=1,\dots,q$ and  $\mathbb{E}(d^2 X_1^2) = \tfrac{1}{\gamma_1} \mathbb{E}(X_1^4) + (p-1)\gamma_1$.
$\mathbb{E}(X_1^4) $ can be easily expressed since $X_1 \sim (1-\epsilon) Z_1+ \epsilon Z_2 $ with $Z_1 \sim  \mathcal{N}(0, 1)$ and $Z_2\sim \mathcal{N}(0,  \alpha)$. Then, we  apply the same properties as in Proof~\ref{proof2} and so we obtain $\mathbb{E}(X_1^4)= 3(1+\epsilon(\alpha^2-1))$. Finally, $\mathbb{E}(d^2 X_1^2) = \tfrac{3}{\gamma_1} (1+\epsilon(\alpha^2-1)) + (p-1)\gamma_1$.

\noindent{All the other $p-q$ diagonal terms are equal to $\mathbb{E}(d^2 X_j^2) =  \mathbb{E}(d^2 X_{q+1}^2) = \mathbb{E}(X_j^4) + (p-1)$ for $j=q+1,\dots,p$. Since $X_{q+1} \sim \mathcal{N}(0,1)$, $\mathbb{E}(X_{q+1}^4) = 3$ and so,  $\mathbb{E}(d^2 X_{q+1}^2) = 3 + (p-1) = p+2$.}

\noindent{\textsl{Computation of $\cov(\bo F_{\bo X})^{-1}\cov_4(\bo F_{\bo X})$:}\\}
Now we can express $\cov(\bo F_{\bo X})^{-1}\cov_4(\bo F_{\bo X})$ as:
$$\cov(\bo F_{\bo X})^{-1}\cov_4(\bo F_{\bo X}) = \frac{1}{(p+2)} \begin{pmatrix}
\tfrac{1}{\gamma_1}\mathbb{E}(d^2 X_1^2) \bo I_q& \bo 0\\
\bo 0 &
 \tfrac{1}{\gamma_2}\mathbb{E}(d^2 X_{q+1}^2) \bo I_{p-q}\end{pmatrix} $$
So, the eigenvalues of $\cov(\bo F_{\bo X})^{-1}\cov_4(\bo F_{\bo X})$ are simply its diagonal terms and the vector space spanned by the eigenvectors associated with $\mathbb{E}(d^2 X_1^2)/ \gamma_1$ is the one spanned by $\{\bo e_1,\ldots,\bo e_q\}$.\\
If $\tfrac{1}{\gamma_1}\mathbb{E}(d^2 X_1^2) > \mathbb{E}(d^2 X_{q+1}^2)$ then $\rho_1(\bo F_{\bo X}) = \dots =  \rho_q(\bo F_{\bo X}) > \rho_{q+1}(\bo F_{\bo X})=\dots= \rho_p(\bo F_{\bo X})$.
This condition is equivalent to: $ \tfrac{3}{\gamma_1^2} (1+\epsilon(\alpha^2-1)) + (p-1)>   p+2 \Leftrightarrow \tfrac{1}{\gamma_1^2} (1+\epsilon(\alpha^2-1)) >  1$ with $\gamma_1 = (1-\epsilon)+\alpha \epsilon$. It leads to the following inequality: $(1-\alpha)^2(1-\epsilon)>0$ which is true for $\alpha>1$ and so the outliers are always revealed on the first $q$ components, as long as $\alpha>1$. If $\alpha=1$ then all the eigenvalues are equal and ICS fails to detect the structure of outlierness.
\end{proof}
%This last result could be extended to the case where the block $\alpha \bo I_q$ in the definition of $\bs \Sigma_5$ is replaced by any diagonal matrix.

\begin{rem}
Given the affine equivariance of ICS, the result can be generalized to the following mixture model where $\bs \Sigma_{51}$  (resp. $\bs \Sigma_{52}$) is a $q\times q$ (resp. $(p-q)\times (p-q)$) matrix not necessarily diagonal:
\begin{equation*}
\bo X \sim (1-\epsilon) \, \mathcal{N}\left(\bs \mu \left(\begin{array}{cc}
\bs \Sigma_{51} & 0 \\
 0 & \bs \Sigma_{52}
 \end{array}
 \right)
 \right) + \epsilon \, \mathcal{N}\left(\bs \mu, \left(\begin{array}{cc}
\alpha \bs \Sigma_{51} & 0 \\
 0 & \bs \Sigma_{52}
 \end{array}
 \right)
 \right)
\end{equation*}
with $\epsilon <0.5$, $\mu$ is a $p$ vector, $q<p$ and $\alpha >1$.

%Moreover, the eigenvectors of $\cov^{-1}(\bo F_{\bo X})\cov_4(\bo F_{\bo X}) $ associated with the $q$ largest eigenvalues span the subspace by the contamination component.
\end{rem}

\end{document}